\documentclass[journal,draftcls,onecolumn,12pt,twoside]{IEEEtran}

\IEEEoverridecommandlockouts 

\usepackage{cite}

\usepackage{soul}
\usepackage{comment}
\usepackage{amssymb, amsthm, amsmath, graphicx, enumerate, cases}
\usepackage{array}
\usepackage{changepage}
\usepackage{xcolor}
\usepackage{balance} 
\usepackage{multirow,booktabs}
\usepackage{diagbox}
\usepackage[normalem]{ulem}
\usepackage{algorithm}
\usepackage{algorithmic}

\DeclareMathOperator{\lcm}{lcm}

\newtheorem{theorem}{Theorem}

\newtheorem{remark}{Remark}
\newtheorem{proposition}{Proposition}
\newtheorem{corollary}{Corollary}
\newtheorem{definition}{Definition}
\newtheorem{construction}{Construction}
\newtheorem{conjecture}{Conjecture}
\newtheorem{example}{Example}

\newcommand{\M}{{\boldsymbol{M}}}
\newcommand{\m}{{\boldsymbol{m}}}

\begin{document}

\title{Multiset Combinatorial Gray Codes with Application to Proximity Sensor Networks}

\author{Chung Shue Chen,~\IEEEmembership{IEEE Senior Member}, Wing Shing Wong,~\IEEEmembership{IEEE Life Fellow}, Yuan-Hsun Lo,~\IEEEmembership{IEEE~Member}, Tsai-Lien Wong
\thanks{Part of this work was presented at the International Symposium on Information Theory and Its Applications (ISITA) in Taipei, Taiwan, 2024~\cite{ISITA24}.} 
\thanks{The research was partially funded by the National Science and Technology Council of Taiwan under Grants NSTC 114-2628-M-153-001-MY3 and NSTC 113-2115-M-110-003-MY2. 
\textit{(Corresponding author: Yuan-Hsun Lo)}
}
\thanks{C. S. Chen is with Nokia Bell Labs,  Paris-Saclay Center, 12 rue Jean Bart, Massy 91300, France. Email: chung\_shue.chen@nokia-bell-labs.com}
\thanks{W. S. Wong is with the Department of Information Engineering, The Chinese University of Hong Kong (CUHK), Shatin, Hong Kong SAR. Email: wswong@ie.cuhk.edu.hk}
\thanks{Y.-H. Lo is with the Department of Applied Mathematics, National Pingtung University, Taiwan.  Email: yhlo0830@gmail.com}
\thanks{T.-L. Wong is with Department of Applied Mathematics, National Sun Yat-sen University, Taiwan. Email: tlwong@math.nsysu.edu.tw}
} 

\maketitle

\begin{abstract}
We investigate coding schemes that map source symbols into multisets of an alphabet.  
Such a formulation of source coding is an alternative approach to the traditional
framework and is inspired by an object tracking problem over proximity sensor networks. 
We define a \textit{multiset combinatorial Gray code} as a mulitset code with fixed multiset cardinality that possesses
combinatorial Gray code characteristic. 
For source codes that are organized as a grid, namely an integer lattice, we propose a solution by first constructing a  mapping from the grid to the set of symbols, which we referred to as colors. 
The codes are then defined as the images of rectangular blocks in the grid of fixed dimensions. 
We refer to the mapping as a \textit{color mapping} and the code as a \textit{color multiset code}.  
We propose the idea of product multiset code that enables us to construct codes for high dimensional grids based on 1-dimensional (1D) grids.  
We provide a detailed analysis of color multiset codes on 1D grids, focusing on codes that require the minimal number of colors.
To illustrate the application of such a coding scheme, we consider an object tracking problem on 2D grids
and show its efficiency, which comes from exploiting transmission parallelism.
Some numerical results are presented to conclude the paper.
\end{abstract}


\section{Introduction}
\label{section:introduction}

By definition, a source code maps symbols in an information source to a set of finite-length strings of symbols from a $k$-ary alphabet.  
In this work, we study a special class of codes, whose images can also be interpreted as 
multisets constructed from the alphabet elements. 
Recall that a \textit{multiset} is a set of elements allowing multiplicity \cite{multiset89} and the order of the elements does not matter. 
So for a multiset code based on the alphabet set $\{A,B\}$ for example, the codewords $AAB$, $ABA$ and $BAA$ are all identical as they have the same number of $A$’s and $B$’s (i.e., these multisets $\{A,A,B\}$, $\{A,B,A\}$ and $\{B,A,A\}$ are identical). 

Since a multiset with $k$ distinct symbols can be represented as a $k$-tuple of integers, a multiset code can also be analyzed in terms of symbol strings.  
However, viewed in this traditional framework, the codewords may be imposed with intertwined constraints.  
For example, a multiset cardinality condition corresponds to a condition on the component sum the $k$-tuple. 
In this paper we will show that multiset can offer a more natural setting to address certain applications.

\textit{Combinatorial Gray codes} are generalization of the binary reflected Gray.  
They map successive source symbols to codewords that differ in ``some prespecified, small way''\cite{CGC97}.
To illustrate with an example, suppose the source symbols are organized as a 1-dimensional (1D) cyclic integer
lattice, which we identify as  {${\mathbb{Z}}_M\triangleq\{0,1,...,M-1\}$}, the ring of residues modulo $M$ for any positive integer $M\geq 2$.
It is natural to consider $i$ and $i+1$ as successive symbols.
For each $m, 1 \leq m<M$, define an $\it m$-block  at $i$ to be $(i,i+1,...,i+m-1)$, where {$i\in {\mathbb{Z}}_M$}.  
Note that there is a one-to-one and onto correspondence between $i$ and the $m$-block it is tagged at. 

If we color each point in  {${\mathbb{Z}}_M$} by a color from a given set of colors $\{c_1,...,c_k\}$, we can obtain a multiset code by mapping each grid point to the color multiset of the {1D \it m}-block it is tagged at. 
In such a coding scheme, code symbols of successive blocks can differ by at most two elements, counting multiplicity. 
Hence, the code can be viewed as a combinatorial Gray code. 
We refer to a multiset code with fixed multiset cardinality that satisfies the combinatorial Gray code description as a \textit{multiset combinatorial Gray code} (MCGC). 
Note that Gray codes whose underlying elements are multisets have also been considered in~\cite{RS96}.

Obviously, the above 1D example to construct an MCGC can be extended to higher dimensional integer lattices, which we refer to as grids. Grids can be used to model proximity sensor networks \cite{BinarySensor95,BinarySensor05,TT11}. 
These wireless micro-sensors, also called binary proximity sensors, would report a target’s presence or absence in their vicinity during object tracking: each sensor outputs a 1 when the target 
is within its sensing range, and 0 otherwise. 
The above sensor network is also known as a \textit{binary sensor network} \cite{BSN03}. 
We get no other information about the location, direction, or other attributes of the target.
This simple model is of fundamental and also practical interest~\cite{Nokia22}, allowing for inexpensive sensing as well as minimal communication. 
For example, ambient Internet-of-Things (IoT) is a wireless sensor network connecting a large number of low-cost self-powered sensor nodes for detecting moving objects for instance under an automated factory or modern warehouse environment. 
Each sensor is equipped with a transmitter that can transmit at a limited data rate to forward its identification number (ID) in order to report the presence of the object. 
One would like to detect the occurrence of the object and track its location. 
This tracking feature is fundamental for industrial IoT as well as many other application scenarios such as smart cities, environment monitoring, logistics and supply chain~\cite{NBIoT17,OTSN20,wild2021,Patent25}.  

In this paper, we formalize an approach to construct MCGCs by first organizing the source symbols as an $n$-dimensional grid and generalize the idea stated for the 1D example.  
We refer to the code as {\it color multiset code}, or {\it color code} for short. 
The construction of color codes depends critically on how an $n$-dimensional grid is colored.  
To reduce the solution complexity, we introduce the idea of a \textit{product multiset code}, which allows high dimension solutions be synthesized from 1D solutions.

In addition to application to proximity sensor networks, constructing efficient MCGC for 1D grids has intriguing connection with {\it Eulerian circuits}, {\it universal cycles} and other important combinatorial concepts \cite{CDG92,HJZ09,GJKO20}. 
Most constructive solutions in the literature, such as universal cycles~\cite{HJZ09}, de Bruijin sequences \cite{Etzion84} and M-sequences \cite{Kumar92}, mainly focused on the scenario that the ordering of the colors of the $m$-block matters.
For more information on de Bruijin sequences, M-sequences and their 2D generalizations, please refer to~\cite{Paterson94,Bruckstein12,CETV25}.

For coding efficiency considerations, it is desirable to
construct MCGCs by using the minimal number of colors.
These codes are referred to as minimal codes.
In this paper, we provide an extensive study on minimal
1D MCGCs. 

To illustrate the application potential of MCGCs, we present a math model for the 2D tracking of an object that moves over a $2$-dimensional (2D) proximity sensor network. 
Indeed, the original inspiration of MCGCs came from this object tracking problem.

The technical results and main contributions of the paper are summarized below.
\begin{itemize} 
    \item We review the concept of combinatorial Gray code to define MCGCs and propose an algorithm for constructing multiset codes by means of defining a color mapping on a high-dimensional grid.
    \item We propose the concept of a product multiset code that allows the construction of high-dimensional MCGCs from 1D MCGCs.
    \item We provide an extensive study on 1D color multiset codes that require the minimum number of colors. 
    The answer depends critically on the cardinality of the color image set.  Explicit solutions are derived for some small value cases. We propose a synthetic construction for general cases, together with an asymptotic analysis on the minimum number of colors.
    \item We show how MCGCs can be applied to object tracking over a 2D proximity sensor network.  
    The MCGCs provide a simple way to arrange for parallelism in data transmission, which results in channel efficiency gains over a single channel transmission.
    \item We study the decoding problem and propose a highly efficient decoding method for a special class of 1D MCGCs. For general high-dimensional MCGCs, if their structure arises from product codes, we also develop a feasible decoding strategy.
    \item We analyze the resource efficiency achieved by using MCGCs in encoding, where the resource is measured by the number of symbols used. An asymptotic result is established to quantify the benefit.
\end{itemize}

The rest of the paper is organized as follows.
In Section~\ref{section:2D-Map-Coloring-Problem}, we formalize the concept of color multiset codes based on
high-dimensional grids and introduce the concept of product multiset code, which reduces the construction of higher-dimensional case to 1D case.
In Section~\ref{section:main-results}, we provide a detailed study of minimal 1D codes as well as
other codes based on code synthesis that can be shown to be asymptotically minimal as the grid size tends to infinity.
In Section~\ref{section:Application}, we discuss an application of MCGCs to object tracking in a 2D
proximity sensor network.
The decoding issue is addressed in Section~\ref{section:decoding}.
We show the performance gain of the proposed protocol against the conventional protocol by deriving the reduction factor in Section~\ref{section:discussion}. 
Finally, a conclusion is presented in Section~\ref{section:conclusion}.

\section{Color Mapping Problem}
\label{section:2D-Map-Coloring-Problem}

\subsection{Mathematical Definition}
\label{section:math-definition}

Let $\mathbb{Z}^+$ denote the set of all positive integers.
For $n\in\mathbb{Z}^+$, let $Z_n$ be the set $\{0,1,\ldots,n-1\}$. 
Note that $\mathbb{Z}_n$ is different from $Z_n$, as the former one refers to the ring of residues modulo $n$ while the latter one just collects all its elements.

Let $n\in\mathbb{Z}^+$.
For an $n$-tuple $\boldsymbol{M}=(M_1,\ldots,M_n)$, where $M_i\in\mathbb{Z}^+$ for all $i$, define an $n$-dimensional integer lattice by $\mathcal{G}_{\M} \triangleq Z_{M_1}\times\cdots\times Z_{M_n}=\{(x_1,\ldots,x_n):\,x_i\in Z_{M_i},\forall i\}$.
For simplicity, we refer an integer lattice as a \textit{grid} in this paper.
Throughout this paper, boldface is used to denote vectors or $n$-tuples.

Given an $n$-tuple $\m=(m_1,\ldots,m_n)$ with $m_i\leq M_i$ for each $i$, define an \textit{$\m$-block} of $\mathcal{G}_{\M}$ as follows.
For $0\leq x_i<M_i-m_i$, $1\leq i\leq n$, the $\m$-block at $(x_1,\ldots,x_n)$ of $\mathcal{G}_{\M}$ is defined to be the set of grid points:
\begin{align*}
\{(x_1+t_1,\ldots,x_n+t_n):\,0\leq t_i<m_i, \forall i\}.
\end{align*}
The \textit{$\m$-coding area} of $\mathcal{G}_{\M}$ is the subset consisting of grid points: $\{(x_1,\ldots,x_n):\,0\leq x_i<M_i-m_i,\forall i\}$.

An \textit{$n$-dimensional color mapping}, $\Phi$, maps $\mathcal{G}_{\M}$ to a set of $k$ colors for some $k\in\mathbb{Z}^+$.
For convenience, we use $[k]\triangleq\{1,2,\ldots,k\}$ to indicate the set of colors. 
Denote by $\mathcal{C}_{\M;k}$ the collection of all $n$-dimensional color mappings on $\mathcal{G}_{\M}$ with $k$ colors. 

Let $\mathcal{P}(e,k)$ represent the collection of multi-subsets of $[k]$ with exactly $e$ elements.

\begin{definition}
Given a color mapping $\Phi\in\mathcal{C}_{\M;k}$ and an $\m$-coding area, an \emph{$n$-dimensional color multiset code} defined by $\Phi$ is a mapping from the coding area to $\mathcal{P}(\prod_{i=1}^n m_i,k)$ so that a point $(x_1,\ldots,x_n)$ is represented by the multiset 
\begin{align*}
    S_{\m}(x_1,\ldots,x_n)\triangleq\{\Phi(x_1+t_1,\ldots,x_n+t_n):\,0\leq t_i<m_i, \forall i\}.
\end{align*}
We refer to $S_{\m}(x_1,\ldots,x_n)$ as the \emph{color multiset tagged at} $(x_1,\ldots,x_n)$.
\end{definition}

An $n$-dimensional color multiset code defined by $\Phi\in\mathcal{C}_{\M;k}$ is called $\m$-\textit{distinguishable} if the multisets $S_{\m}(x_1,\ldots,x_n)$ are all distinct for all grid points $(x_1,\ldots,x_n)$ in the $\m$-coding area.
In other words, if we associate a grid point in the coding area by the $\m$-block it is tagged at and identify it by the multiset of colors the block points mapped to, then the grid points are uniquely identified.

A color mapping problem aims to find a distinguishable color multiset code for a given grid $\mathcal{G}_\M$ and a block size $\m$. 
Note that two distinct blocks tagged at $\boldsymbol{x}=(x_1,\ldots,x_n)$ and $\boldsymbol{x'}=(x'_1,\ldots,x'_n)$ can be considered as neighbors if $|\boldsymbol{x}-\boldsymbol{x'}|\triangleq \sum_{i=1}^n|x_i-x'_i| =1$, namely, they are differ by exactly one coordinate with difference $1$.
Under a color multiset code, the two color multisets tagged at two neighboring points, each of which contains $\prod_{i=1}^{n}m_i$ elements, can differ by at most  
\begin{equation}\label{eq:differ-elements}
2\Big(\prod_{i=1}^{n}m_i - (m_j-1)\prod_{i\neq j}m_i\Big) = 2\prod_{i\neq j}m_i
\end{equation}
elements, counting multiplicity, where $j$ indicates the unique distinct coordinate index.
Note also that the equation in Eqn.~\eqref{eq:differ-elements} is defined to be $2$ in the case when $n=1$, namely, the 1D case.
Hence, one can view such a code as a combinatorial Gray code.

When it comes to 1D and 2D cases, the two basic cases $n=1,2$, we will avoid redundant parentheses in notation if it does not cause any ambiguity.
That is, we simply use $\mathcal{G}_M$, $m$-block, $\mathcal{C}_{M;k}$, $S_m(x)$ for 1D case, and $\mathcal{G}_{M_1,M_2}$, $(m_1,m_2)$-block, $\mathcal{C}_{M_1,M_2;k}$, $S_{m_1,m_2}(x,y)$ for 2D case.

\begin{example}\label{ex:1D-2D} \rm
Fig.~\ref{fig:1D} shows a color mapping on a 1D grid of size $10$ using $4$ colors.
We use black to denote the indices of the grid points, and red to denote the color labels.
When $m=3$, the color multiset code defined by it has: $S_3(0)=\{1,1,2\},S_3(1)=\{1,2,3\},\ldots,S_3(7)=\{2,4,4\}$.
So it defines a $3$-distinguishable color multiset code. 
\begin{figure}[h]
\centering
\includegraphics[width=0.6\columnwidth]{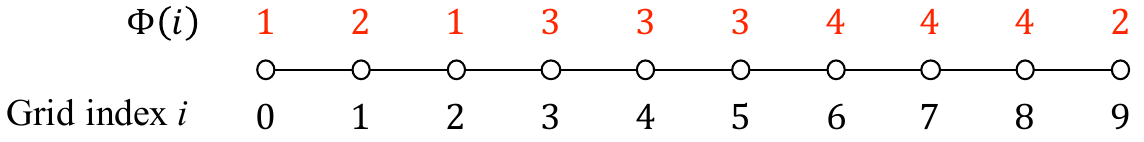}
    \caption{A 1D color mapping on $\mathcal{G}_{10}$ with $4$ colors.}  \label{fig:1D} 
\end{figure} 

Fig.~\ref{fig:2D} shows a color mapping on a 2D grid of size $7\times 5$ using $4$ colors.
The color label is displayed at the upper-right corner of the corresponding grid point.
When $m_1=m_2=2$ (i.e., block size $2\times 2$), we have $S_{2,2}(0,0)=\{1,1,1,1\}, S_{2,2}(0,1)=\{1,1,2,3\}$, and so on.
One can easily verify that this color mapping defines a $(2,2)$-distinguishable multiset color code.
\begin{figure}[h]
\centering
\includegraphics[width=0.35\columnwidth]{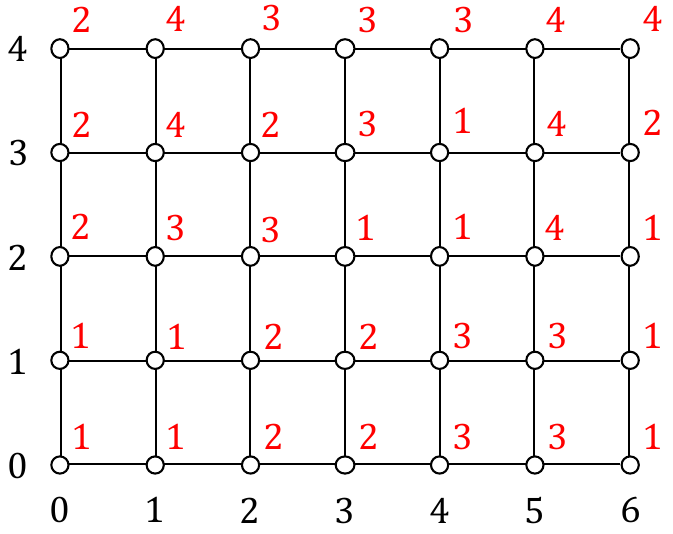}
    \caption{A 2D color mapping on $\mathcal{G}_{5,7}$ with $4$ colors.}  \label{fig:2D} 
\end{figure} 
\end{example}

In some applications, one may wish to identify the grid point $(x_1,\ldots,M_j,\ldots,x_n)$ with $(x_1,\ldots,0,\ldots,x_n)$, $1\leq j\leq n$, so that the coding area assumes the geometric characteristics of an $n$-dimensional torus.
(For example, one may allow the color pattern be repeated after a certain distance similar to radio frequency reuse in a wireless communication system.)  
In this case, we view the grid as a {\it cyclic grid} and denote it by $\mathcal{G}^c_{\M}$.
The \textit{cyclic} version of the color multiset $S_{\m}(\boldsymbol{x})$ can be defined for any point $\boldsymbol{x}=(x_1,\ldots,x_n)\in\mathcal{G}^c_{\M}$ accordingly using corresponding modular arithmetic.
In this sense, the cyclic grid $\mathcal{G}^c_{\M}$ is isomorphic to $\mathbb{Z}_{M_1}\times\cdots\times\mathbb{Z}_{M_n}$ and the coding area is equal to the whole grid.  
Moreover, if $m_i<M_i$ for all $i$, then each $\m$-block is tagged at a unique point.
Let $\mathcal{C}^c_{\M;k}$ denote the collection of all $n$-dimensional color mappings on $\mathcal{G}^c_{\M}$ with $k$ colors. 
An $n$-dimensional color multiset code defined by $\Phi\in\mathcal{C}^c_{\M;k}$ is cyclic $\m$-distinguishable if the multisets $S_{\m}(\boldsymbol{x})$ are all distinct for all $\boldsymbol{x} \in\mathcal{G}^c_{\M}$.
The 1D color mapping in Example~\ref{ex:1D-2D} defines a cyclic $3$-distinguishable color multiset code, while the 2D case is not a cyclic $(2,2)$-distinguishable since $S_{2,2}(5,0)=\{1,1,3,3\}=S_{2,2,}(6,0)$.

For efficiency considerations, it is natural to seek color multiset codes that require the minimal number of colors.   
We denote by $K_{\M}(\m)$ the minimum number of colors required for the existence of an $\m$-distinguishable color multiset code on $\mathcal{G}_{\M}$, and $K^c_{\M}(\m)$ for the cyclic case. 
The 1D and 2D cases are denoted simply by $K_M(m), K^c_M(m)$ and $K_{M_1,M_2}(m_1,m_2), K^c_{M_1,M_2}(m_1,m_2)$, respectively.

\subsection{Product Multiset Code}
\label{section:product-code}

One can construct $n$-dimensional color mappings based on 1D solutions. 
There are multiple algorithms that can achieve this.
Here we provide a simple approach based on the idea of {\it product multiset code}.
Note that, in~\cite{ISITA24}, the construction of 1D color multiset codes was investigated and solutions
were found for certain parameters only.
We will provide further results on this in Section~\ref{section:main-results}.

Consider positive integers $M_1,\ldots,M_n$.
For $1\leq i\leq n$, let $\Phi_i\in\mathcal{C}_{M_i;k_i}$.
Consider a new set of colors with $\prod_{i=1}^n k_i$ elements and index its elements as $n$-tuples in $[k_1]\times\cdots\times[k_n]$.
We define an $n$-dimensional color mapping $\Phi$ on $\mathcal{G}_{\M}$ with $\prod_{i=1}^n k_i$ colors by setting
\begin{equation}\label{eq:product-code-def}
\Phi(x_1,\ldots,x_n)=(\Phi_1(x_1),\ldots,\Phi_n(x_n))
\end{equation}
for $(x_1,\ldots,x_n) \in\mathcal{G}_{\M}$.

\begin{proposition}\label{prop:product-code}
The product multiset code defined by $\Phi$ given in Eqn.~\eqref{eq:product-code-def} is $\m$-distinguishable if and only if the code defined by $\Phi_i\in\mathcal{C}_{M_i;k_i}$ is $m_i$-distinguishable for all $i$.
\end{proposition}
\begin{proof}
The ``only if'' part is obvious by definition, so it suffices to consider the ``if'' part.

Based on the $n$-tuple representation, for any multi-subset $B\subseteq[k_1]\times\cdots\times[k_n]$, we define the first coordinate projection of $B$, $\mathsf{P}_1(B)$, by:
\begin{align*}
\mathsf{P}_1(B)\triangleq\{i_1:(i_1,i_2,\ldots,i_n)\in B \}.
\end{align*}
Note that the set is defined as a multiset so that $\mathsf{P}_1(B)$ always containing the same number of elements as $B$, i.e., $|\mathsf{P}_1(B)|=|B|$.

Now suppose the color multiset of a block tagged at $\boldsymbol{x}=(x_1,\ldots,x_n)$, $S_{\m}(\boldsymbol{x})$, is equal to that of another block tagged at $\boldsymbol{x'}=(x'_1,\ldots,x'_n)$, $S_{\m}(\boldsymbol{x'})$. 
We shall show that $x_i=x'_i$ for all $i$.

We identify the projected elements, $\mathsf{P}_1(S_{\m}(\boldsymbol{x}))$, by first focusing on grid points of the form $(x,\hat{x}_2,\ldots,\hat{x}_n)$ for fixed coordinates $\hat{x}_2,\ldots,\hat{x}_n$, where $x_j\leq \hat{x}_j<x_j+m_j$ for $j=2,\ldots,n$. 
Let
\begin{align*}
\mathcal{D}_{(\hat{x}_2,\ldots,\hat{x}_n)}\triangleq\{c:\,(c,\Phi_2(\hat{x}_2),\ldots,\Phi_n(\hat{x}_n))=\Phi(x_1+t,\hat{x}_2,\ldots,\hat{x}_n), 0\leq t < m_1\}.
\end{align*}
By Eqn.~\eqref{eq:product-code-def}, the multiset $\mathcal{D}_{(\hat{x}_2,\ldots,\hat{x}_n)}$ is independent of $\hat{x}_2,\ldots,\hat{x}_n$ and we simply represent it as $\mathcal{D}$.
Hence, $\mathsf{P}_1(S_{\m}(\boldsymbol{x}))$, which contains $\prod_{i=1}^n m_i$ elements, is equal to the union of $\prod_{i=2}^n m_i$ copies of $\mathcal{D}$. 
Note also, $\mathcal{D}$ is the multiset that codes the $m_1$-block in $\mathcal{G}_{M_1}$ tagged at $x_1$ under $\Phi_1$.

Similarly, for the grid point $\boldsymbol{x'}$, define for $\hat{x}'_2,\ldots,\hat{x}'_n$, $x'_j\leq \hat{x}'_j<x'_j+m_j$ for $j=2,\ldots,n$, 
\begin{align*}
\mathcal{E}\triangleq\{c:\,(c,\Phi_2(\hat{x}'_2),\ldots,\Phi_n(\hat{x}'_n))=\Phi(x'_1+t,\hat{x}'_2,\ldots,\hat{x}'_n), 0\leq t<m_1\}.
\end{align*}
The set is independent of $\hat{x}'_2,\ldots,\hat{x}'_n$ and is well-defined.
It is equal to the color multiset code of the $m_1$-block in $\mathcal{G}_{M_1}$ tagged at $x'_1$ under $\Phi_1$.
Moreover, $\mathsf{P}_1(S_{\m}(\boldsymbol{x'}))$ is equal to the union of $\prod_{i=2}^{n}m_i$ copies of $\mathcal{E}$.

By assumption, the two projected multisets are equal, so it follows that $\mathcal{D}=\mathcal{E}$.
Given that $\Phi_1$ is $m_1$-distinguishable, we have $x_1=x'_1$. 

Following the same argument, one can show that $x_j=x'_j$ for $j=2,\ldots,n$.
Hence the proof is completed.
\end{proof}


The following result is a direct consequence of Proposition~\ref{prop:product-code}. 

\begin{proposition}\label{prop:minimal}
Consider two $n$-tuples $\M=(M_1,\ldots,M_n)$ and $\m=(m_1,\ldots,m_n)$ for some $M_1,\ldots,M_n,m_1,\ldots,,m_n\in\mathbb{Z}^+$ with $m_i\leq M_i$ for all $i$. 
One has
\begin{equation}    
K_{\M}(\m)\leq \prod_{i=1}^n K_{M_i}(m_i).
\end{equation}
\end{proposition}

We will focus on finding solutions to the 1D color mapping problem in the next section.

\section{1D Code Construction and Sequence Length Bounds} \label{section:main-results}

The 1D color mapping problem is to seek for the minimal number of colors, $k=K_M(m)$, such that an $m$-distinguishable color code on $\mathcal{G}_M$ exists, for given $m$ and $M$.
In this section, we consider an equivalent problem: For given $k$ and $m$, maximize the grid size $M$ such that an $m$-distinguishable color code on $\mathcal{G}_M$ with $k$ colors exists.
Precisely speaking, let $M_m(k)$ denote the maximum value $M$ such that an $m$-distinguishable color code on $\mathcal{G}_M$ with $k$ colors exists.
We have
\begin{align*}
K_M(m) &= \min\{k:\,M_m(k)\geq M\}, \text{ and} \\
M_m(k) &= \max\{M:\,K_M(m)\leq k\}.
\end{align*}
In other words, $K_M(m)$ can be determined if the values $M_m(k)$ are known, and vice versa.

After introducing some necessary preliminaries and a general upper bound on the maximum value $M$ in Section~\ref{section:math-modeling}, we obtain a general lower bounds by the help of some combinatorial structures in the literature in Section~\ref{section:lower-bounds}.
Section~\ref{section:construction-method} is devoted to derive the explicit values for some small $m$.
Finally, in Section~\ref{section:recursive-construction}, we will propose a synthetic construction for larger $m$, together with an asymptotic analysis on the maximum value $M$.

\subsection{Mathematical Model for 1D Coloring}
\label{section:math-modeling}

Let $M,m,k\in\mathbb{Z}^+$ with $m\leq M$.
A 1D color mapping $\Phi\in\mathcal{C}_{M;k}$ can be realized as a sequence $S=s_0s_1\cdots s_{M-1}$ by letting $s_i=\Phi(i)$ for $i\in Z_M$.
In such a fashion, the color multiset tagged at point $t$ is represented as $S_m(t)\triangleq\{s_t,s_{t+1},\ldots,s_{t+m-1}\}$, and we say the sequence $S$ is \textit{$m$-distinguishable} if all $S_m(t)$ are distinct for all $0\leq t\leq M-m$.
Similarly, $S=s_0s_1\cdots s_{M-1}$ can also be used to represent a 1D color mapping on a cyclic grid $\mathcal{G}^c_M(=\mathbb{Z}_M)$, and is called \textit{cyclic $m$-distinguishable} if all $S_m(t)$, $t\in\mathbb{Z}_M$, are distinct.
For fixed $m$ and $k$, we use $M^c_m(k)$ to denote the maximum length of a cyclic $m$-distinguishable sequence whose elements (colors) are in $[k]$.

\begin{example}\rm
Let $M=15$ and $k=5$, and consider the sequence $S=1,2,2,3,4,4,5,1,1,3,3,5,5,\allowbreak 2,4$.
When $m=2$, it is easy to see that $S$ is both $2$-distinguishable and cyclic $2$-distinguishable; when $m=3$, $S$ is $3$-distinguishable but not cyclic $3$-distinguishable since both $S_{3}(13),S_{3}(14)$ consist of the three integers $1,2,4$.
\end{example}

Let $S=s_0,s_1,\ldots,s_{M-1}$ be a cyclic $m$-distinguishable sequence.
For $t\in\mathbb{Z}_M$, define the $t$\textit{-cut} of $S$ as the string
\begin{equation}\label{eq:sequence-cut}
s_t,s_{t+1},\ldots, s_{M-1},s_0,s_1,\ldots, s_{t-1}, s_{t}, s_{t+1},\ldots, s_{t+m-2}.
\end{equation}
That is, ``cutting'' the original cyclic sequence $S$ at the position $s_{t}$ and repeating the consequent $m-1$ elements.
It is easy to see that each of the color multisets $S_m(t)$, $t\in\mathbb{Z}_M$, appears exactly once as a multiset of some consecutive $m$ elements of the string. 
Hence the $t$-cut is an $m$-distinguishable sequence.
Take $S=1,1,1,2,2,2,3,3,3$, a cyclic $3$-distinguishable sequence, as an example. 
The $5$-cut of $S$, say $2,3,3,3,1,1,1,2,2,2,3$, is a $3$-distinguishable sequence.

The following proposition is from the $t$-cut action.

\begin{proposition}\label{prop:t-cut}
If there is a cyclic $m$-distinguishable sequence on $[k]$ of length $M$, then $M_m(k)\geq M+m-1$. 
In particular,
\begin{equation}\label{eq:M_m-inequality}
M_m(k)\geq M^c_m(k)+m-1.
\end{equation}
\end{proposition}

We define a useful notation 
\begin{equation}\label{eq:nonnegative-H}
    H^k_m\triangleq\binom{k+m-1}{m},
\end{equation}
which stands for the number of solutions of non-negative integers to the equation $x_1+x_2+\cdots+x_k=m$.
Note that $H^k_m$ is also the cardinality of the set $\mathcal{P}(m,k)$, where the notation is given in Section~\ref{section:math-definition}.

The following upper bounds on $M^c_{m}(k)$ and $M_{m}(k)$ is directly from the definition. 

\begin{proposition}\label{prop:upper}
For given positive integers $m$ and $k$, one has
\begin{equation} \label{eq:upper-cyclic}
M^c_{m}(k) \leq \binom{k+m-1}{m}
\end{equation}
and
\begin{equation} \label{eq:upper}
M_{m}(k) \leq \binom{k+m-1}{m}+m-1. 
\end{equation}
\end{proposition}
\begin{proof}
We only consider Eqn.~\eqref{eq:upper}, since the cyclic version can be dealt with in the same way.
Let $S=s_0s_1\cdots s_{M-1}$ be a longest $m$-distinguishable sequence on $[k]$.
By definition, the $M-m+1$ multisets $S_m(t)$, $0\leq t\leq M-m$, are all distinct.
By representing each of these multiset as $\{1^{e_1},\ldots,k^{e_k}\}$, where $e_s$ indicates the multiplicity of the element $s$, those non-negative multiplicities must satisfy $e_1+e_2+\cdots+e_k=m$.
By Eqn.~\eqref{eq:nonnegative-H}, the number of all possible $S_m(t)$ is $\binom{k+m-1}{m}$.
It follows that $M-m+1\leq \binom{k+m-1}{m}$.
\end{proof}

\begin{remark}\rm 
The concept of cyclic $m$-distinguishable sequences was recently discussed in~\cite{ICALP24} and generalized to higher dimensional cases, called \textit{grid colourings}.
The key structure, \textit{vector sum packing}, of the proposed solution is a cyclic $m$-distinguishable sequence with an assumption that the sums of any $m$ consecutive elements are all distinct. 
The construction method and result neither affect nor cover our subsequent findings.
\end{remark}

\subsection{Bounds of $M_m(k)$ Based on Previously Known Results}
\label{section:lower-bounds}

A cyclic $m$-distinguishable sequence on $[k]$ in which every multiset in $\mathcal{P}(m,k)$ appears exactly once as a color multiset is known as an $(m,k)$-\textit{Mcycle}~\cite{HJZ09}.
In other words, an Mcycle is precisely a cyclic $m$-distinguishable sequence whose length attains the theoretical upper bound given in~Eqn.~\eqref{eq:upper-cyclic}.
It should be noted that the requirement for a (cyclic) $m$-distinguishable sequence is merely that all color multisets are distinct; it does not demand that all possible combinations appear.

By definition, any permutation of $[k]$ is a $(1,k)$-Mcycle.
Let $K_k$ denote the complete graph of $k$ vertices labeled by elements in $[k]$.
An \textit{Eulerian circuit} of a graph is a circuit that contains all edges.
See Fig.~\ref{fig:Eulerian}\,(a) for an example of an Eulerian circuit of $K_5$.
Let $G_k$ denote the graph obtained from $K_k$ by adding a self-loop at each vertex.
Obviously, the list of vertices traveled by an Eulerian circuit of $G_k$ is a $(2,k)$-Mcycle.
As a graph contains an Eulerian circuit if and only if each vertex's degree is even, a $(2,k)$-Mcycles exists for all odd $k$.
See \cite[Theorem 1.2.26]{West} for more details about Eulerian circuits.

For general $m$, the sufficient condition of the existence of an $(m,k)$-Mcycle is that $m$ divides $\binom{k+m-1}{m}$.
It was conjectured that the necessary part holds for any case if $k$ is sufficiently large.

\begin{conjecture}[\cite{HJZ09}]\label{conj:Mcycle}
For all $m$ there is an integer $y_0(m)$ such that, for $k\geq y_0(m)$, an $(m,k)$-Mcycle exists if and only if $m$ {divides} $\binom{k+m-1}{m}$.
\end{conjecture}

Our previous discussion verifies that Conjecture~\ref{conj:Mcycle} holds for $m=1,2$.
Here is the most up-to-date result on this conjecture.

\begin{theorem}[\cite{HJZ09}]\label{thm:Mcycle-known}
Let $y_0(3)=4, y_0(4)=5$ and $y_0(6)=11$.
Then, for $m\in\{3,4,6\}$ and $k\geq y_0(m)$, an $(m,k)$-Mcycle exists whenever $k$ is relatively prime to $m$, i.e., $\gcd(m,k)=1$.
\end{theorem}

We immediately have the following result on $M_m(k)$ by Propositions~\ref{prop:t-cut}--\ref{prop:upper},  Theorem~\ref{thm:Mcycle-known} and aforementioned arguments on $m=1,2$ cases.

\begin{corollary}\label{corol:Mcycle} 
$M_1(k)=k$ for all $k$, and $M_2(k)=\binom{k+1}{2}+1$ for all odd $k$.
For $m\in\{3,4,6\}$, $k\geq y_0(m)$ with $\gcd(m,k)=1$, where $y_0(3),y_0(4),y_0(6)$ are given in Theorem~\ref{thm:Mcycle-known}, the following holds:
\begin{align*}
    M_m(k) = \binom{k+m-1}{m}+m-1.
\end{align*}
\end{corollary}

For the cases missing in Corollary~\ref{corol:Mcycle} when $m=2,3,4,6$, we have the following lower bounds.

\begin{corollary}\label{corol:Mcycle-lower}
Let $y_0(2)=2$, $y_0(3)=4, y_0(4)=5$ and $y_0(6)=11$. The following holds:
\begin{enumerate}[(i)]
\item For $m\in\{2,3,4\}$, $k\geq n_o(m)$ and $\gcd(m,k)\neq 1$, we have
    \begin{align*}
    M_m(k) \geq \binom{k+m-2}{m}+2m-1.
    \end{align*} 
\item For $m=6$, $k\geq y_0(6)$, we have
\begin{align*}
M_6(k) \geq 
\begin{cases}
    \binom{k+4}{6}+11 & \text{if }k\equiv 0,2\, (\bmod\, 6), \\
    \binom{k+3}{6}+17 & \text{if }k\equiv 3\, (\bmod\, 6), \\
    \binom{k+2}{6}+23 & \text{if }k\equiv 4\, (\bmod\, 6).
\end{cases}
\end{align*}
\end{enumerate}
\end{corollary}
\begin{proof}
(i) For each case, we have $\gcd(m,k-1)=1$.
By Theorem~\ref{thm:Mcycle-known}, there exists an $(m,k-1)$-Mcycle, in which elements are all in $[k-1]$.
Then, the result follows by appending $m$ consecutive ``$k$'' at the end of any $c$-cut of the $(m,k-1)$-Mcycle.

(ii) The $m=6$ cases can be dealt with similarly by finding the smallest $t$ such that $\gcd(6,k-t)=1$.
The resulting sequence is obtained by appending $m$ consecutive ``$k$'', $m$ consecutive ``$k-1$'', down to $m$ consecutive ``$k-t+1$'', at the end of any $c$-cut of the $(m,k-t)$-Mcycle, in which elements are all in $[k-t]$.
The result hence follows by $t=1$ for $k\equiv 0,2$, $t=2$ for $k\equiv 3$, and $t=3$ for $k\equiv 4$ (mod $6$).
\end{proof}

Table~\ref{tab:M_m(k)-bounds} collects all known $M_m(k)$ based on Proposition~\ref{prop:upper}, and Corollaries~\ref{corol:Mcycle}, \ref{corol:Mcycle-lower}. 
Symbol ``$\star$'' means the corresponding lower bound matches the theoretical upper bound.

\begin{table}[h]
    \centering
    \begin{tabular}{|c|c|c|c|}
    \hline
       \multirow{2}{*}{$m$} & \multirow{2}{*}{$k$ (mod $m$)} 
        & lower bound & upper bound \\
        & & of $M_m(k)$ & of $M_m(k)$ \\ \hline 
        \multirow{2}{*}{$2$} & $1$ & $\star\binom{k+1}{2}+1$ & \multirow{2}{*}{$\binom{k+1}{2}+1$} \\ 
            & $0$ & $\binom{k}{2}+3$ & \\ \hline
        \multirow{2}{*}{$3$} & $1,2$ & $\star\binom{k+2}{3}+2$ & \multirow{2}{*}{$\binom{k+2}{3}+2$} \\
            & $0$ & $\binom{k+1}{3}+5$ & \\ \hline
        \multirow{2}{*}{$4$} & $1,3$ & $\star\binom{k+3}{4}+3$ & \multirow{2}{*}{$\binom{k+3}{4}+3$} \\ 
            & $0,2$ & $\binom{k+2}{4}+7$ & \\ \hline
        \multirow{4}{*}{$6$} & $1,5$ & $\star\binom{k+5}{6}+5$ & \multirow{4}{*}{$\binom{k+5}{6}+5$} \\ 
            & $0,2$ & $\binom{k+4}{6}+11$ &  \\
            & $3$ & $\binom{k+3}{6}+17$ &  \\
            & $4$ & $\binom{k+2}{6}+23$ &  \\ \hline
    \end{tabular}
    \medskip
    \caption{Known values and bounds on $M_m(k)$ for $m=2,3,4$ and $6$.}
    \label{tab:M_m(k)-bounds}
\end{table}

Another combinatorial structure in the literature related to cyclic $m$-distinguishable sequences is the \textit{universal cycles}~\cite{CDG92}, \textit{Ucycles} for short.
An $(m,k)$-Ucycle is a cyclic $m$-distinguishable sequence $S$ in which there is no repeated elements in any $S_m(t)$ and every $m$-subset of $[k]$ appears exactly once as a $S_m(t)$ for some $t$.
Unlike an Mcycle, a Ucycle does not admit multisets.
Note that an $(m,k)$-Ucycle is of length $\binom{k}{m}$.
Similar to Mcycles, it was conjectured in~\cite{CDG92} that for any $m$ there is an integer $y_0(m)$ such that, $k\geq y_0(m)$, an $(m,k)$-Ucycle exists if and only if $m$ divides $\binom{k-1}{m-1}$. 
The conjecture has been proved in~\cite{GJKO20} using probabilistic methods; however, a constructive approach has yet to be found.
See~\cite{GJKO20,Zhang20,Hurlbert94,CHHM09,BG11} for more information about Ucycles and related topics.

\begin{remark} \rm
Theorem~\ref{thm:Mcycle-known} is based on a construction of Ucycles given in~\cite{Hurlbert94}.
The construction relies on a special structure, called \textit{good $d$-patterns}, that enables an $(m,k)$-Ucycle to be built recursively in a systematic manner.
The existence of this structure is characterized precisely by $m=3,4$ or $6$ and $\gcd(m,k)=1$, which is the reason for the exclusion of the case $m=5$ in Theorem~\ref{thm:Mcycle-known} and Corollaries~\ref{corol:Mcycle} and \ref{corol:Mcycle-lower}.
\end{remark}

\subsection{New Cyclic $m$-distinguishable Sequences}
\label{section:construction-method}

As an $m$-distinguishable sequence can be obtained from a cyclic $m$-distinguishable sequence, i.e., Proposition~\ref{prop:t-cut}, from now on we will study $M_m^c(k)$ in more details.

Recall that an $(m,k)$-Mcycle exists only when $m$ divides $\binom{k+m-1}{m}$.
Let $p$ be a prime.
Observe that $k$ divides $\binom{k+p-1}{p}$ if and only if $k$ is not divisible by $p$.
This indicates that when the block size $m=p$ and $k$ is divisible by $p$, the prior known requirement of the existence of a $(p,k)$-Mcycle does not hold, and thus $M^c_p(k)$ is strictly less than $\binom{k+p-1}{p}$.
Our first task in this subsection is to derive a tighter upper bound in this case.

\begin{theorem}\label{thm:new-upper-bound}
Suppose $p$ is a prime and $k$ is divisible by $p$. 
Then the following holds:
\begin{align*}
    M^c_p(k) \leq \binom{k+p-1}{p}-\frac{k}{p}.
\end{align*}
\end{theorem}
\begin{proof}
Assume $M=M^c_p(k)$.
Let $S=s_0s_1\cdots s_{M-1}$ be a longest cyclic $p$-distinguishable sequence on $[k]$.
Fix an element $a\in[k]$.
For a multiset $A$ on $[k]$, let $\varphi_a(A)$ denote the number of $a$'s in $A$.

Denote by $\mathcal{B}_a$ the collection of all $p$-multisets of $[k]$ that contain at least one element $a$.
Let us count the number of appearances of $a$'s in $\mathcal{B}_a$.
For $0\leq i\leq p-1$, let $\mathcal{B}_{a,i}\subset\mathcal{B}_a$ denote the set of $p$-multisets containing exactly $(p-i)$ $a$'s, i.e.,
\begin{align*}
    \mathcal{B}_{a,i}=\{A\in\mathcal{B}_a:\,\varphi_a(A)=p-i\}.
\end{align*}
Observe that $|\mathcal{B}_{a,i}|=H^{k-1}_i$, that is, the number of solutions of non-negative integers to the equation $e_1+\cdots+e_{a-1}+e_{a+1}+\cdots e_k=i$, where $e_s$ indicates the multiplicity of the element $s$.
Then, 
\begin{align*}
    \sum_{A\in\mathcal{B}_{a,i}}\varphi_a(A)=(p-i)H^{k-1}_i.
\end{align*}
It is not hard to see from Eqn.~\eqref{eq:nonnegative-H} that $H^{k-1}_i$ is a multiple of $p$ for $i=0$ and $2\leq i\leq p-1$.
When $i=1$, $\sum_{A\in\mathcal{B}{a,_1}}\varphi(A)=(p-1)(k-1)$, which is not divisible by $p$ due to the assumption that $p$ divides $k$. 
As $\mathcal{B}_a$ is the disjoint union of $\mathcal{B}_{a,i}$, $i=0,1,\ldots,p-1$, the value $\sum_{i=0}^{p-1}\sum_{A\in\mathcal{B}_{a,i}}\varphi_a(A)$ is not divisible by $p$, that is,
\begin{equation}\label{eq:new-UB-1}
\sum_{A\in\mathcal{B}_a}\varphi_a(A) \text{ is not divisible by } p.
\end{equation}

Let us consider the longest cyclic $p$-distinguishable sequence $S$.  
Let $n_a$ be the number of $a$'s in $S$ and let $t_i$ be the places for the appearance $a$ in $S$, for $i=1,2, \cdots, n_a$. 
Recall $S_p(t)$, $0\leq t\leq M-1$, denotes the $p$-multiset $\{s_t,s_{t+1},\ldots,s_{t+p-1}\}$, where the indices are taken modulo $M$.
For any fixed $t_i$, the element $a$ located at $t_i$ is contained in exactly $p$ $p$-multisets, say $S_p(t_i-p+1)$, $S_p(t_i-p+2)$, $\ldots, S_p(t_i)$.
Summing up $t_i$ for $i=1,2, \cdots, n_a$, it follows that 
\begin{equation}\label{eq:new-UB-2}
    \sum_{t=0}^{M-1}\varphi_a(S_p(t)) = \sum_{i=1}^{n_a}p = n_ap.
\end{equation}


Eqns.~\eqref{eq:new-UB-1}--\eqref{eq:new-UB-2} conclude that the set of $M$ $p$-multisets produced by $S$, say $\bigcup_{t=0}^{M-1}S_p(t)$, can not include all $p$-multisets in $\mathcal{B}_a$, for $a=1,2,\ldots,k$.
Therefore, $\bigcup_{t=0}^{M-1}S_p(t)$ must exclude at least $k/p$ multisets, and thus the result follows.
\end{proof}

In what follows, we will construct cyclic $m$-distinguishable sequences on $[k]$ with $m=2,3$ and $m|k$, where the sequence lengths achieve the upper bounds given in Theorem~\ref{thm:new-upper-bound}.

In graph theory, a \textit{1-factor} of a graph is a spanning 1-regular subgraph, i.e., a collection of $k/2$ independent edges, where $n$ is the number of vertices.

\begin{theorem}\label{thm:m=2}
For any even $k\geq 2$, there exists a cyclic 2-distinguishable sequence on $[k]$ of length $\binom{k+1}{2}-\frac{k}{2}$.
\end{theorem}
\begin{proof}
Let $F$ be a $1$-factor in $K_k$.
Then, $K_k-F$ is an even graph, i.e., each vertex's degree is even.
A $2$-distinguishable sequence can be obtained by the following three steps.
\begin{enumerate}[(i)]
    \item Find an Eulerian circuit in $K_k-F$.
    \item Repeat the first occurrence of every element and denote by $S$ the obtained cyclic sequence.
    \item Pick an arbitrary $t$-cut of $S$.
\end{enumerate}

Take Fig.~\ref{fig:Eulerian}\,(b) as an example, where the Eulerian circuit is set to $1,3,2,4,5,1,6,2,5,3,6,4$.
Following the three steps, we have
\begin{align*}
1,3,2,4,5,1,6,2,5,3,6,4 &\stackrel{\text{(ii)}}{\longrightarrow} 1,1,3,3,2,2,4,4,5,5,1,6,6,2,5,3,6,4 \\ &\stackrel{\text{(iii)}}{\longrightarrow} 1,1,3,3,2,2,4,4,5,5,1,6,6,2,5,3,6,4,1,
\end{align*}
where the last step is done by taking the $0$-cut.

It is not hard to see that the resulting sequence $S$ following these three steps is $2$-distinguishable, which is of length $\binom{k}{2}-\frac{k}{2}+k+1$, as desired.
\end{proof}

\begin{figure}
\centering
\includegraphics[width=0.45\columnwidth]{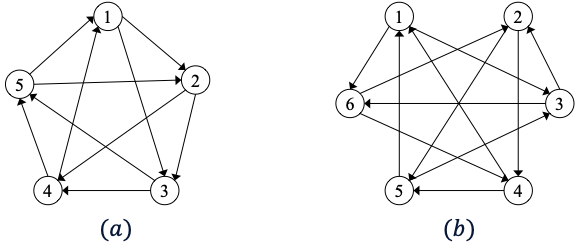}
    \caption{(a) An Eulerian circuit of $K_5$. The vertices that it travels is $1,2,3,4,5,1,3,5,2,4$ in order. (b) An Eulerian circuit of $K_6-F$, where $F$ contains the three edges $\{1,2\}, \{3,4\}$ and $\{5,6\}$. The vertices that it travels is $132451625364$ in order.}  \label{fig:Eulerian} 
\end{figure} 

In the case of $m=3$ and $k\equiv 0$ (mod $3$), we shall provide a recursive construction based on the method given in~\cite[Theorem 8]{HJZ09}.
The main idea is to construct a cyclic $3$-distinguishable sequence as longer as possible, and then apply the $t$-cut action.
The following example shows two $3$-distinguishable sequences, for $k=3$ and $k=6$, which will stand for the initial cases for the recursive construction.

\begin{example}\label{ex:Mpacking_36} \rm
The sequence
\begin{equation}\label{eq:Mpacking_3}
    1,1,1,2,2,2,3,3,3
\end{equation}
is a cyclic $3$-distinguishable sequence which contains all $3$-multisets of $[3]$ except the one $\{1,2,3\}$.
And the sequence
\begin{equation} \label{eq:Mpacking_6}
    \begin{split}
        &1,1,1,2,2, 2,3,3,3,1, \ 1,6,6,3,1, 5,5,2,2,4, \ 5,3,5,3,2, 4,4,3,3,6, \\
        &2,1,4,1,4, 6,2,6,2,5, \ 1,4,3,6,5, 5,5,4,4,4, \ 6,6,6,5
    \end{split}
\end{equation}
is a cyclic $3$-distinguishable sequence which contains all $3$-multisets of $[6]$ except the two $\{1,2,3\}, \{4,5,6\}$.
Therefore, one has $M^c_3(3)\geq 9$ and $M^c_3(6)\geq 54$. 
\end{example}

\begin{theorem}\label{thm:m=3}
For $k$ a multiple of $3$, there exists a cyclic $3$-distinguishable sequence on $[k]$ of length $\binom{k+2}{3}-\frac{k}{3}$.
\end{theorem}
\begin{proof}
The proof is proceeds by induction on $\frac{k}{3}$.
As shown in Example~\ref{ex:Mpacking_36}, the assertion holds when $k=3$ and $k=6$, so we will consider $k\geq 9$.
In the inductive construction, the obtained cyclic $3$-distinguishable sequence on $[k]$ will be of the form $ST$, where $ST$ is the concatenation of two subsequences $S$ and $T$ satisfying the following properties.
\begin{enumerate}
    \item[(S1)] $S$ is a cyclic $3$-distinguishable sequence on $[k-3]$ and contains all $3$-multisets of $[k-3]$ except the $\frac{k}{3}-1$ ones: $\{1,2,3\},\ldots,\{k-5,k-4,k-3\}$.
    \item[(S2)] $S$ begins with $1,1$ if we consider it is a non-cyclic sequence.
    \item[(T1)] $T$ is a cyclic $3$-distinguishable sequence on $[k]$ and contains all $3$-multisets of $[k]$ with at least one element from $\{k-2,k-1,k\}$ but excludes the one $\{k-2,k-1,k\}$.
    \item[(T2)] $T$ begins with $1,1$ and ends with $k,k-1$ if we consider it is a non-cyclic sequence.
\end{enumerate}
Note that we use $ST$ to denote the concatenation of two subsequences $S$ and $T$, rather than $S,T$, in order to emphasize the concatenation itself.

With the four properties, $ST$ is a cyclic $3$-distinguishable sequence on $[k]$ containing all $3$-multisets of $[k]$ except the $\frac{k}{3}$ ones: $\{1,2,3\},\ldots,\{k-5,k-4,k-3\},\{k-2,k-1,k\}$, and thus is of length $\binom{k+2}{3}-\frac{k}{3}$.
Notice that the cyclic $3$-distinguishable sequence on $[6]$ shown in Eqn.~\eqref{eq:Mpacking_6} satisfies the four properties by letting $S=1,1,1,2,2,2,3,3,3$ and $T$ the remaining subsequence.

Assume as induction hypothesis that there exists a cyclic $3$-distinguishable sequence on $[k-3]$ of the form $S'T'$, where the subsequence $S'$ satisfies the two conditions S1 and S2, and $T'$ satisfies the two conditions T1 and T2.
That is, $S'$ is a cyclic $3$-distinguishable sequence on $[k-6]$ and contains all $3$-multisets of $[k-6]$ except the $\frac{k}{3}-2$ ones: $\{1,2,3\},\ldots,\{k-8,k-7,k-6\}$, and $T'$ is a cyclic $3$-distinguishable sequence on $[k-3]$ and contains all $3$-multisets of $[k-3]$ with at least one element from $\{k-5,k-4,k-3\}$ but excludes the one $\{k-5,k-4,k-3\}$.
By viewing $S'$ and $T'$ as non-cyclic subsequences, $S'$ begins with $1,1$ and $T'$ also begins with $1,1$ and ends with $k-3,k-4$.

Now, we shall construct a cyclic $3$-distinguishable sequence on $[k]$ of the form $ST$, where the subsequence $T$ will be of the form $T=XYZ$, the concatenation of three subsequences $X,Y$ and $Z$.

First, let $S=S'T'$.
Obviously, $S$ satisfies conditions S1 and S2.

Next, let $X$ be a sequence obtained from $T'$ by replacing each $k-5$ by $k-2$, $k-4$ by $k-1$, and $k-3$ by $k$.
Notice that $X$ begins with $1,1$ and ends with $k,k-1$ due to the structure of $T'$.
Moreover, by viewing $X$ as a non-cyclic sequence, it contains all $3$-multisets of $[k-6]\cup\{k-2,k-1,k\}$ with at least one element from $\{k-5,k-4,k-3\}$ but excludes the three multisets $\{k-2,k-1,k\}$, $\{1,k-1,k\}$ and $\{1,1,k-1\}$.

Finally, the constructions of both $Y$ and $Z$ {are divided into two cases} according to the parity of $k$. 
For notational convenience, we use symbols $a,b,c,d,e,f$ to denote elements $k-5,k-4,k-3,k-2,k-1,k$, respectively.
When $k$ is even, let 
\begin{align*}
Y=\text{a, a, f, f, c, a, e, e, b, b, \ d, e, c, e, c, b, d, d, c, c, \ f, b, a, d, a, d, f, b, f},
\end{align*}
and 
\begin{align*}
Z=\ & \text{b, e, } k-6, \text{ a, f, } k-7, \text{ b, e, } k-8, \text{ a, f, } k-9, \ldots, \text{a, f, } 1, \text{b, e}, \\
    & \text{a, d, } k-6, \text{ c, e, } k-7, \text{ a, d, } k-8, \text{ c, e, } k-9, \ldots, \text{c, e, } 1, \text{a, d}, \\
    & \text{c, f, } k-6, \text{ b, d, } k-7, \text{ c, f, } k-8, \text{ b, d, } k-9, \ldots, \text{b, d,  } 1, \text{c, f, e}.
\end{align*}
One can check that the non-cyclic sequence $YZ$ contains all $3$-multisets of $\{k-5,k-4,k-3,k-2,k-1,k\}$ except the two $\{k-5,k-4,k-3\},\{k-2,k-1,k\}$ and also contains all $3$-multisets with one element from each of $[k-6]$, $\{k-5,k-4,k-3\}$, and $\{k-2,k-1,k\}$.
Note that the multisets $\{1,k-1,k\}$ and $\{1,1,k-1\}$ are missing in the non-cyclic sequence $X$, but both will appear in the concatenation of $Z$ and $S$.
Therefore, $ST=SXYZ$ is the desired cyclic $3$-distinguishable sequence which satisfies the four properties S1, S2, T1 and T2.

The case when $k$ is odd can be dealt with in a similar way, so we just list the constructions of the corresponding subsequences $Y$ and $Z$, as follows.
\begin{align*}
Y = \ & \text{b, e, b, 1, f, a, b, d, 1, c, \ f, f, a, a, e,  c, b, f, b, f}, \\ 
& \text{d, a, d, a, 1,  e, c, c, f, a, \ e, e, c, d, c, d, b, d},
\end{align*}
and
\begin{align*}
Z=\ & \text{b, e}, k-6, \text{a, f}, k-7, \text{b, e}, k-8, \text{a, f}, k-9, \ldots, \text{a, f}, 2, \text{b, e}, \\
    & \text{a, d}, k-6, \text{c, e}, k-7, \text{a, d}, k-8, \text{c,e}, k-9, \ldots, \text{c, e}, 2, \text{a, d}, \\
    & \text{c, f}, k-6, \text{b, d}, k-7, \text{c, f}, k-8, \text{b, d}, k-9, \ldots, \text{b, d}, 2, \text{c, f, e}.
\end{align*}
\end{proof}

We use the following example to illustrate the construction given in Theorem~\ref{thm:m=3}.

\begin{example}\label{ex:Mpacking_9} \rm
When $k=9$, we pick Eqn.~\eqref{eq:Mpacking_6} as a cyclic $3$-distinguishable sequence on $[6]$, $S=S'T'$, where $S'=1,1,1,2,2,2,3,3,3$ and 
\begin{align*}
T'= & 1,1,6,6,3, 1,5,5,2,2, \ 4,5,3,5,3, 2,4,4,3,3, \ 6,2,1,4,1, 4,6,2,6,2, \\ 
    & 5,1,4,3,6, 5,5,5,4,4, \ 4,6,6,6,5.
\end{align*}
Then, by replacing each $4,5,6$ in $T'$ by $7,8,9$, respectively, we get
\begin{align*}
X = & 1,1,9,9,3, 1,8,8,2,2, \ 7,8,3,8,3, 2,7,7,3,3, \ 9,2,1,7,,1, 7,9,2,9,2, \\ 
    & 8,1,7,3,9, 8,8,8,7,7, \ 7,9,9,9,8. 
\end{align*}
Since $k=9$ is odd, we have
\begin{align*}
Y = \ & 5,8,5,1,9, 4,5,7,1,6, \ 9,9,4,4,8, 6,5,9,5,9, \\ 
    & 7,4,7,4,1, 8,6,6,9,4, \ 8,8,6,7,6, 7,5,7,
\end{align*}
and 
\begin{equation*}
Z = 5,8,3,4,9,2,5,8,\ 4,7,3,6,8,2,4,7,\ 6,9,3,5,7,2,6,9,8.
\end{equation*}
The resulting cyclic $3$-distinguishable sequence on $[9]$ is $ST=S'T'XYZ$, which is of length $9+45+45+38+25=162=\binom{11}{3}-3$.
\end{example}

We immediately have the following result.

\begin{corollary}\label{corol:Mcycle-m=23}
For any even $k\geq 2$,
\begin{align*}
    M^c_2(k)=\binom{k+1}{2}-\frac{k}{2} \text{ and } M_2(k)\geq\binom{k+1}{2}-\frac{k}{2}+1.
\end{align*}
For any $k$ with $3|k$,
\begin{align*}
    M^c_3(k)=\binom{k+2}{3}-\frac{k}{3} \text{ and } M_3(k)\geq\binom{k+2}{3}-\frac{k}{3}+2.
\end{align*}
\end{corollary}
\begin{proof}
The proof is directly followed by Theorems~\ref{thm:new-upper-bound}--\ref{thm:m=3} and Proposition~\ref{prop:t-cut}.
\end{proof}

\vspace{0.2cm}
\begin{remark}\rm
The recursive construction in the proof of Theorem~\ref{thm:m=3} can be generalized to $m\geq 4$, by concatenating two sequences which satisfy some specific properties.
More precisely, suppose there exists a cyclic $m$-distinguishable sequence on $[k-m]$ in the form $S'T'$. 
It is desired to obtain a cyclic $m$-distinguishable sequence on $[k]$ in the form $ST$, where $S=S'T'$ and $T=XW$, having the following properties. 
\begin{enumerate}[(1)]
    \item The first $m-1$ elements of $S'$ an $T'$ are identical and all in $[k-2m]$.
    \item The last $m-1$ elements of $T'$ are all in $\{k-2m+1,k-2m+2,\ldots,k-m\}$.
    \item $X$ is obtained from $T'$ by replacing each element $k-2m+t$ by $k-m+t$ for $t=1,2,\ldots,m$.
    \item The last $m-1$ elements of $W$ are identical with those of $X$.
\end{enumerate} 
\end{remark}

\subsection{A Synthetic Construction}
\label{section:recursive-construction}

Complexity for constructing (cyclic) $m$-distinguishable sequences increases with the value $m$ and so far we have only discussed explicit construction algorithms for small $m$ only.
In what follows, we shall provide an approach to construct cyclic $m$-distinguishable sequences by splitting $m$.

\medskip

\noindent\textbf{Synthetic Construction.}
Suppose $M_1,M_2,m_1,m_2$ are positive integers such that $m_i$ divides $M_i$ and $\gcd(d,M_i/(m_id))=1$ for $i=1,2$, where $d=\gcd(M_1/m_1,M_2/m_2)\geq 2$.
Let $S=s_0,s_1,\ldots, s_{M_1-1}$ be a cyclic $m_1$-distinguishable sequence on $[k_1]$ and $T=t_0,t_1,\ldots, t_{M_2-1}$ a cyclic $m_2$-distinguishable sequence on $\{k_1+1,k_2+2,\ldots,k_1+k_2\}$, a set of $k_2$ colors.
We evenly divide $S$ into $M_1/m_1$ subsequences, called \emph{$\alpha$-words}, in the form $S=\alpha_0\alpha_1\cdots\alpha_{(M_1/m_1)-1}$, where $\alpha_j=s_{m_1j},s_{m_1j+1},\ldots, s_{m_1j+m_1-1}$ for $j\in\mathbb{Z}_{M_1/m_1}$.
In other words, $S$ can be viewed as the concatenation of subsequences $\alpha_0,\alpha_1,\ldots,\alpha_{(M_1/m_1)-1}$.
Similarly, divide $T$ into \emph{$\beta$-words} in the form $T=\beta_0\beta_1\cdots\beta_{(M_2/m_2)-1}$, where $\beta_j=t_{m_2j},t_{m_2j+1},\ldots, t_{m_2j+m_2-1}$ for $j\in\mathbb{Z}_{M_2/m_2}$.
Let $L\triangleq\lcm(M_1/m_1,M_2/m_2)=M_1M_2/(m_1m_2d)$.
Define the \textit{cross product} of $S$ and $T$, denoted by $S\times T$, as a sequence of length $(m_1+m_2)L=\frac{M_1M_2}{d}(\frac{1}{m_1}+\frac{1}{m_2})$ by
\begin{equation}\label{eq:cross-product-def}
    S\times T=\alpha_0\beta_0\alpha_1\beta_1\cdots\alpha_i\beta_i\cdots\alpha_{L-1}\beta_{L-1},
\end{equation}
where the indices of $\alpha$-words (resp., $\beta$-words) are taken modulo $M_1/m_1$ (resp., $M_2/m_2$).

We must emphasize once more that our notation $\alpha_i\beta_j$ stands for the concatenation of two subsequences $\alpha_i$ and $\beta_j$, rather than a sequence containing only the two elements (symbols) $\alpha_i$ and $\beta_j$.
This clarification is made to avoid any possible misunderstanding.
Note that when both $\alpha_i$ and $\beta_j$ contain only one element, the expressions $\alpha_i\beta_j$ and $\alpha_i, \beta_j$ are equivalent in meaning.

\begin{example}\label{ex:cross-product}\rm
Let $M_1=12,M_2=30,m_1=2,m_2=3$.
We have $M_1/m_1=6, M_2/m_2=10$, and thus $d=\gcd(M_1/m_1,M_2/m_2)=2\geq 2$, $\gcd(d,M_i/(m_id))=1$ for $i=1,2$.
Consider $k_1=k_2=5$.
We pick the cyclic $2$-distinguishable sequence of length $12$ as 
\begin{align*}
S=\underbrace{1,1,}_{\alpha_0}\ \underbrace{3,3,}_{\alpha_1}\ \underbrace{5,2,}_{\alpha_2}\ \underbrace{4,1,}_{\alpha_3}\ \underbrace{2,3,}_{\alpha_4}\ \underbrace{4,5,}_{\alpha_5}
\end{align*}
and the cyclic $3$-distinguishable sequence of length $30$ as
\begin{align*}
T=\underbrace{6,6,6,}_{\beta_0}\ \underbrace{7,7,7,}_{\beta_1}\ \underbrace{8,8,8,}_{\beta_2}\ \underbrace{9,9,9,}_{\beta_3}\ \underbrace{10,10,10,}_{\beta_4}\ \underbrace{6,6,8,}_{\beta_5}\ \underbrace{8,10,10,}_{\beta_6}\ \underbrace{7,9,6,}_{\beta_7}\ \underbrace{8,10,7,}_{\beta_8}\ \underbrace{7,9,9}_{\beta_9}.
\end{align*}
The cross product of $S$ and $T$ is
\begin{align*}
S\times T = & \ \alpha_0\beta_0\alpha_1\beta_1\alpha_2\beta_2\alpha_3\beta_3\alpha_4\beta_4\alpha_5\beta_5 \\
& \ \alpha_0\beta_6\alpha_1\beta_7\alpha_2\beta_8\alpha_3\beta_9\alpha_4\beta_0\alpha_5\beta_1 \\
& \ \alpha_0\beta_2\alpha_1\beta_3\alpha_2\beta_4\alpha_3\beta_5\alpha_4\beta_6\alpha_5\beta_7 \\
& \ \alpha_0\beta_8\alpha_1\beta_9\alpha_2\beta_0\alpha_3\beta_1\alpha_4\beta_2\alpha_5\beta_3 \\
& \ \alpha_0\beta_4\alpha_1\beta_5\alpha_2\beta_6\alpha_3\beta_7\alpha_4\beta_8\alpha_5\beta_9.
\end{align*}
\end{example}

\begin{theorem}\label{thm:cross-product}
The sequence $S\times T$ constructed in Synthetic Construction is a cyclic $(m_1+m_2)$-distinguishable sequence with $k_1+k_2$ colors of length $\frac{M_1M_2}{d}(\frac{1}{m_1}+\frac{1}{m_2})$.
\end{theorem}
\begin{proof}
The indices of neighboring $\alpha$- and $\beta$-words in Eqn.~\eqref{eq:cross-product-def} can be identified as a set of ordered pairs 
\begin{align*}
\mathcal{O}\triangleq\{(x \bmod{M_1/m_1}, x \bmod{M_2/m_2}):\,0\leq x\leq L-1\}.
\end{align*}
Recall that $d=\gcd(M_1/m_1,M_2/m_2)$ and $L=\lcm(M_1/m_1,M_2/m_2)$.
By the Chinese Remainder Theorem, the system of congruences 
\begin{align*}
    \begin{cases}
        x\equiv i \bmod{M_1/m_1} \\
        x\equiv j \bmod{M_2/m_2}
    \end{cases}
\end{align*}
has a solution in $\mathbb{Z}_L$ if and only if $i\equiv j \bmod{d}$, and the solution is unique.
Therefore,
\begin{equation}\label{eq:cross-product_ordered-index}
\mathcal{O}=\{(i,j):\,i\in\mathbb{Z}_{M_1/m_1},j\in\mathbb{Z}_{M_2/m_2}, \text{ and } i\equiv j \bmod{d}\},
\end{equation}
where the size of the right-hand-side is exactly $L$.

By the structure of $S\times T$ as shown in Eqn.~\eqref{eq:cross-product-def}, each subsequence of length $(m_1+m_2)$ in $S\times T$ consist of consecutive $m_1$ elements in $S$ and consecutive $m_2$ elements in $T$, and at least one of the two segments is an $\alpha$- or $\beta$-word.
Suppose to the contrary that $S\times T$ is not cyclic $(m_1+m_2)$-distinguishable.
Let $X$ and $Y$ be two distinct subsequences of length $(m_1+m_2)$ in $S\times T$ with identical multisets.
Since the color sets in $S$ and $T$ are disjoint, the two subsequences $X,Y$ must have the same $\alpha$- or $\beta$-word.
Without loss of generality, assume that $\alpha_i$ is the common subsequence in $X$ and $Y$ for some $i\in\mathbb{Z}_{M_1/m_1}$.
Let 
\begin{align*}
X=A\alpha_iB \text{ and } Y=C\alpha_iD,
\end{align*}
where $AB$ and $CD$ are some two subsequences of length $m_2$ in $T$.
Since $T$ is cyclic $m_2$-distinguishable, to get a contradiction, it suffices to show that $AB\neq CD$.
We assume that the $\beta$-word followed by $\alpha_i$ in $X$ (resp., $Y$) is $\beta_j$ (resp., $\beta_{j'}$), where $j,j'\in\mathbb{Z}_{M_2/m_2}$.
Note that $B$ is a part of $\beta_j$ and $D$ is a part of $\beta_{j'}$.
By the characterization of the indices of $\alpha$ and $\beta$ words in Eqn.~\eqref{eq:cross-product_ordered-index}, we have
\begin{align*}
    j\equiv i\equiv j' \bmod{d} \text{ and } j\neq j'.
\end{align*}
If $B$ and $D$ are both non-empty, $AB\neq CD$ due to $j\neq j'$.
If both $B$ and $D$ are empty, then $AB=\beta_{j-1}\neq\beta_{j'-1}=CD$ because of $j\neq j'$.
Finally, consider the case when one of $B$ and $D$ is empty.
By symmetry, assume $X=\alpha_i\beta_{j}$ and $Y=\beta_{j'-1}\alpha_i$.
Since $j\neq j'$ and $j\equiv j' \bmod{d}$ with $d\geq 2$, we have $j\neq j'-1$, which concludes that $AB=\beta_j\neq\beta_{j'-1}=CD$.
This completes the proof.
\end{proof}

In previous subsections, we have shown some concrete constructions of cyclic $2$- and $3$-distinguishable sequences for any color number $k$.
Synthetic Construction provides a recursive method to construct a cyclic $m$-distinguishable sequence for any $m$.
For example, we can build a cyclic $10$-distinguishable sequence from two cyclic $5$-distinguishable ones, each of which can be produced by taking the cross product of one cyclic $2$-distinguishable and one cyclic $3$-distinguishable sequence, as shown in Example~\ref{ex:cross-product}.
Theoretically, for any $m$, we can have a cyclic $m$-distinguishable sequence.
Moreover, we have the following tight asymptotic bound of $M^c_m(k)$.

\begin{theorem}\label{thm:asymptotic-cyclic}
    For any $m$ and $k$, by viewing $M^c_m(k)$ as a function of $k$, one has 
    \begin{align*}
        M^c_m(k) = \Theta(k^m).
    \end{align*}
\end{theorem}
\begin{proof}
Eqn.~\eqref{eq:upper-cyclic} implies that $M^c_m(k)=O(k^m)$.
We shall show that $M^c_m(k)=\Omega(k^m)$ by induction on $m$.
The cases when $m=2,3$ can be verified by Corollary~\ref{corol:Mcycle} and Corollary~\ref{corol:Mcycle-m=23}.

Consider $m\geq 4$, and let $m=m_1+m_2$ for some $m_1,m_2\geq 2$.
Assume the assertion holds for all numbers less than $m$, that is, there exist constants $c_1,c_2>0$ such that $M^c_{m_1}(k/2)\geq c_1(k/2)^{m_1}$ and $M^c_{m_2}(k/2)\geq c_2(k/2)^{m_2}$.
By Bertrand-Chebyshev Theorem, there exists a prime $p_1$ with
\begin{align*}
    \frac{c_1(k/2)^{m_1}}{4m_1} < p_1 < \frac{c_1(k/2)^{m_1}}{2m_1}.
\end{align*}
Since $2m_1p_1<M^c_{m_1}(k/2)$, we can pick a cyclic $m_1$-distinguishable sequence on $[k/2]$ of length $M_1=2m_1p_1$.
Similarly, pick a cyclic $m_2$-distinguishable sequence on $[k/2]$ of length $M_2=2m_2p_2$, where $p_2$ is a prime with
\begin{align*}
    \frac{c_2(k/2)^{m_2}}{4m_2} < p_2 < \frac{c_2(k/2)^{m_2}}{2m_2}.
\end{align*}
We may assume $p_1$ and $p_2$ are relatively prime.
To guarantee this, one way is to fine-tune the two constants $c_1,c_2$ so that the two intervals which $p_1$ and $p_2$ belong to are disjoint.
Observe that $d=\gcd(M_1/m_1,M_2/m_2)=2$ and $\gcd(d,M_i/(m_id))=1$ for $i=1,2$.
As $M_1>\frac{c_1}{2}(k/2)^{m_1}$ and $M_2>\frac{c_2}{2}(k/2)^{m_2}$, by Theorem~\ref{thm:cross-product}, the cross product of the above two sequences is a cyclic $m$-distinguishable sequence on $[k]$ of length
\begin{equation}\label{eq:asymptotic-cyclic-proof}
\frac{M_1M_2}{d}\left(\frac{1}{m_1}+\frac{1}{m_2}\right) > \frac{c_1c_2}{2^{m+3}}\left(\frac{1}{m_1}+\frac{1}{m_2}\right)k^m.
\end{equation}
The coefficient of $k^m$ in Eqn.~\eqref{eq:asymptotic-cyclic-proof} is independent of $k$, which leads to $M^c_m(k)=\Omega(k^m)$.
\end{proof}

We end this section with an asymptotic bound of $M_m(k)$, which can be derived immediately by Proposition~\ref{prop:t-cut} and Proposition~\ref{prop:upper}.

\begin{corollary}\label{corol:asymptotic}
    For any $m$ and $k$, by viewing $M_m(k)$ as a function of $k$, one has 
    \begin{align*}
        M_m(k) = \Theta(k^m).
    \end{align*}
\end{corollary}

\section{An Application of MCGCs} 
\label{section:Application}

A basic solution for object tracking in a
proximity sensor network is to assign each sensor node a unique identification number (ID). If the sensor detects an object within a predefined region, it reports the
detection to a remote observer using its own dedicated
communication channel. 
The detection regions are assumed to
be mutually exclusive, so that a single sensor is activated
at any given time. The channels are assumed to have identical data rate.  If the object location needs to
be reported at a fixed rate, this imposes 
a natural bound on the length of the ID bit
length, which increases as the size of the sensor network
grows.

On the other hand, if the communication 
{channels assigned
to the sensors do not} interference among them, one can exploit the parallelism to improve the communication efficiency.
Suppose an object can  simultaneously trigger multiple sensors in its vicinity to transmit a code
symbol, one can then {employ} MCGC to exploit communication
parallelism of the sensors.  

Consider a monitored square area of size $L \times L$.    
An object can randomly appear on it. 
We divide time into discrete time slots of duration $h$. 
At the beginning of each time slot, the object may show up anywhere with $(x,y)$-coordinates, where $x \in [0,L]$ and $y \in [0,L]$. 
We want to determine the position of the object at the beginning of each time slot $t$, for each $t$ with an upside precision $\delta$, i.e., if the system says the object is located at $(i \delta, j \delta)$, then it is located in a block area of $x \in ((i-0.5) \delta, (i+0.5) \delta]$ and $y \in ((j-0.5) \delta, (j+0.5) \delta]$. 
We refer to such a block area as a \textit{basic cell}.

For object tracking and localization, we rely on proximity sensors that can detect and report the presence of an object in its neighborhood.
At each time slot, if the object is located within the predefined neighborhood of a sensor, we assume it can determine which basic cell the object is located at. The technology to achieve such a goal, using sensors that work individually or as a team, is well studied and is not the focus of this paper.  

The primary objective of our sensor network is to track the object and report its occurrence to a remote observer in a timely manner.
To fix idea for subsequent discussion, we introduce the concept of a {\it detection block}.  
If the sensor network employs a block size of $m$, then a sensor deployed at position $(i \delta, j \delta)$ would report the presence of an object that is located at $x \in ((i-m/2) \delta, (i+m/2) \delta]$ and $y \in ((j-m/2) \delta, (j+m/2) \delta]$.
To report the discovery, all sensors are equipped with a transmitter that can transmit $B$ bits per $h$ time unit. So, the data rate is $R = B/h$ bits/s. 
 
 \begin{figure}[t]
    \centering
    \includegraphics[width=0.47\columnwidth]{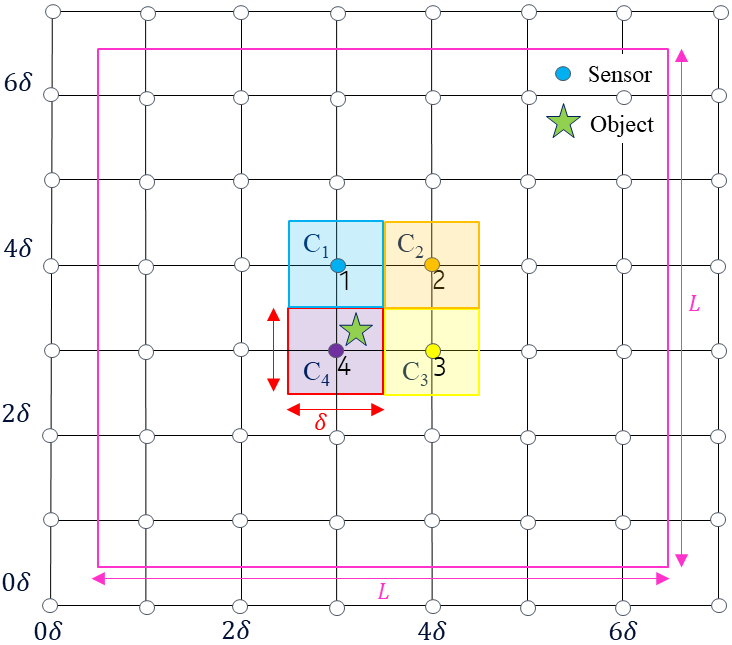}
    \caption{
    $L = 6 \delta$. 
    Sensors are deployed in the grid topology to cover the monitored area $L \times L$, indicated by the magenta square. 
    A sensor at $(i \delta, j \delta)$ would detect the presence of object in the vicinity $x \in ((i-0.5) \delta, (i+0.5) \delta]$ and $y \in ((j-0.5) \delta, (j+0.5) \delta]$, and then transmits the pre-assigned unique ID to signal. In the above example, sensor 4 is triggered and sends its ID, denoted by $\text{C}_4$, to inform a remote observer.}  
    \label{fig:Fig1} 
\end{figure} 

We note that we are using an idealized tracking model as we ignore issues such as overlapping detection and noise errors. 
Techniques for handling noisy sensor detection and imperfections, for example by means of Kalman filter or particle filter can be found in \cite{Filter00,Filter02,Sensors17} and related references. 

\subsection{Baseline Reference}
\label{section:protocol-A}

As a baseline reference solution where no color coding is needed, we consider the special case where block size $m$ is set to be 1.
In other words, 
if an object is located at 
$x \in ((i-0.5) \delta, (i+0.5) \delta]$ and $y \in ((j-0.5) \delta, (j+0.5) \delta]$, 
where $i, j = 1, 2,3, \ldots$, 
only the sensor at 
$(i\delta, j\delta)$ 
is triggered and it reports a detection of an object to the remote observer by transmitting its own ID (identification number or label).\footnote{Sensor sends its ID (a number of bits) when the object enters its detection area or range.} 

Given that each sensor has a communication channel with data rate $R$, since there are $\lceil L/\delta \rceil ^2$ sensors, we require 
\begin{equation}\label{eq:protocol-A}
{\log_2 \left(\lceil L/\delta \rceil^2 \right) \leq Rh,}
\end{equation}
which implies that $L \leq \delta 2^{Rh/2}$. For notational convenience in coming discussions, we 
define {$
{K} \triangleq\lceil L/\delta \rceil^2$}.

Fig.~\ref{fig:Fig1}  illustrates the system and its setup. As an example, an object appears in the coverage area of sensor 4, which is located at $(3 \delta, 3 \delta)$. Therefore, sensor 4 is triggered, which then reports the detection 
by sending its 
ID 
denoted by $\text{C}_4$.
It follows that $
{K}$ 
distinct IDs are required and each ID requires $\log_2 
{K}$ bits for distinct identification. If the 
expected length $L$ is 
greater than $\delta 2^{Rh/2}$ (i.e., Eqn.~\eqref{eq:protocol-A} does not hold), this simple unique ID protocol is not feasible.


\begin{figure}[t]
    \centering
    \includegraphics[width=0.47\columnwidth]{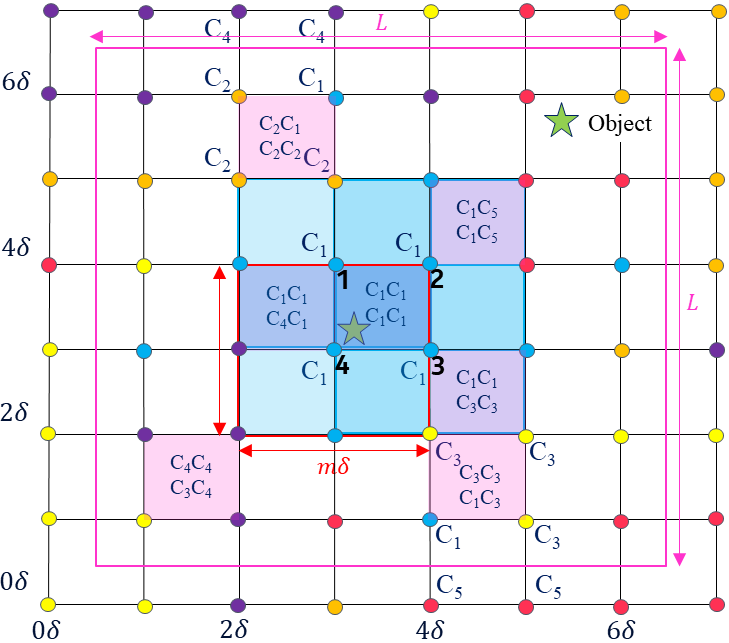}
    \caption{
    $L = 6 \delta$ and $m = 2$. 
    Sensors are deployed for  monitoring the same area with same size $L \times L$, indicated by the magenta square. 
    A sensor at $(i \delta, j \delta)$ would detect the presence of object in the  vicinity $x \in ((i-1) \delta, (i+1) \delta]$ and $y \in ((j-1) \delta, (j+1) \delta]$, and then transmits the pre-assigned ID to signal. 
    Here, sensors 1, 2, 3 and 4 are  triggered as $m = 2$. In contrary to the setup in Fig.~\ref{fig:Fig1}, we do not have to apply distinct ID for each sensor. 
    As shown above, only 5 distinct IDs (indicated as 5 colors) are required for the  sensors since the remote observer can distinguish where the object is located by each set of $m^2$-neighboring sensor IDs. 
    } 
    \label{fig:Fig2b}
\end{figure}

\subsection{MCGC Protocol (Color Coding Protocol)}
\label{section:protocol-B}

A natural alternative is to set $m$ to be strictly larger than 1 so that 
the sensor located at $(i\delta, j\delta)$ 
detects an object located at $(x,y)$ if
{$x \in ((i-m/2) \delta, (i+m/2) \delta]$ and $y \in ((j-m/2) \delta, (j+m/2) \delta]$.} 
An illustration with $m=2$ is provided in 
Fig.~\ref{fig:Fig2b}. 

To report the location of an object for a given localization accuracy, it is not necessary to require each sensor reports a unique ID.
Instead we can use a set of $k$ colors for the 
$K$ sensors, each color only requires $\log_2 k$ bits to represent.  The color mapping must satisfy the condition that the combination of each set of 
$m^2$
``\textit{neighboring}'' labels would be distinguishable 
when being collected by the remote observer.  In other words, we require a MCGC with $k$ symbols and code size {equal to 
$K$}.

The channel data rate condition now becomes: 
\begin{equation}\label{eq:protocol-B}
\log_2 k 
\leq Rh.
\end{equation}

Given $k$ symbols, we are interested in finding MCGCs with the largest code space. 
Using the example in  Fig.~\ref{fig:Fig2b}, 
when the object appears in the area  
$(3 \delta, 4 \delta]$ $\times$ $(3 \delta, 4 \delta]$, 
the sensors 1, 2, 3 and 4, which are located at
$(3\delta, 4\delta)$,  $(4\delta, 4\delta)$, $(4\delta, 3\delta)$ and  $(3\delta, 3\delta)$ respectively, 
are triggered as $m = 2$ and they will report the 
discovery by sending their pre-assigned colors  
to the remote observer. 
In contrast to the baseline protocol in Fig.~\ref{fig:Fig1}, 
the solution in
Fig.~\ref{fig:Fig2b} requires a symbol size of $k = 5$ (instead of 
36)  
since the remote observer can distinguish where the object is located by each combination of 
$m^2$-neighboring IDs 
in the grid. 
For example, $\{\text{C}_1, \text{C}_1, \text{C}_1, \text{C}_1\}$, $\{\text{C}_1, \text{C}_1, \text{C}_1, \text{C}_4\}$, $\{\text{C}_1, \text{C}_1, \text{C}_3, \text{C}_3\}$, $\{\text{C}_3, \text{C}_3, \text{C}_3, \text{C}_1\}$ and $\{\text{C}_4, \text{C}_4, \text{C}_4, \text{C}_3\}$ are distinguishable. 

Note in the setup of our proximity sensor network, 
the remote observer can decode the received signal to read the collected IDs (each set  has $m^2$ elements) but cannot distinguish 
their permutations\footnote{
For example, 
$\{\text{C}_1, \text{C}_1, \text{C}_1, \text{C}_4\}$, 
$\{\text{C}_1, \text{C}_1, \text{C}_4, \text{C}_1\}$, 
$\{\text{C}_1, \text{C}_4, \text{C}_1, \text{C}_1\}$ and 
$\{\text{C}_4, \text{C}_1, \text{C}_1, \text{C}_1\}$
would be considered as equivalent as they are the same multiset, since  they have the same number of $\text{C}_1$'s and the same number of $\text{C}_4$'s.},  
i.e., the ordering of the received signals (IDs) cannot be 
determined. 



It should be noted that a 
sequence construction scheme for a string scanning and recognition system 
was proposed in \cite{Sinden85}.   
It has a sliding window measurement of the scanned string that is similar to our consideration.  However, it is considered that 
the geographical order of the received information in the process could be known. 
That is not our case, thus the above prior art cannot be applied to solve our problem.


Let's denote the number of multiset combinations of the IDs of $m^2$-neighboring 
sensors by $N^{m}(k)$. In order to locate the object, we need: 
\begin{equation}\label{eq:unique-codes-number}
{K} = \lceil L/\delta \rceil^2 \leq N^{m}(k).
\end{equation}
Once $L,\delta$ and $m$ are given, it suffices to find $k$ such that $
{K} \leq N^{m}(k)$. 
It is clear that 
the 
reduction on the number of required distinct IDs (or colors) is from 
{$K$} to $k$.   
An illustration with $m = 3$ is provided in Fig.~\ref{fig:Fig3a}.  
A sensor at   $(i\delta, j\delta)$ 
can detect objects in    
$x \in ((i-m/2) \delta, (i+m/2) \delta]$ and $y \in ((j-m/2) \delta, (j+m/2) \delta]$. 
When the object appears in the area  
$(2.5 \delta, 3.5 \delta]$ $\times$ $(2.5 \delta, 3.5 \delta]$, 
the sensors $1, 2, \ldots, 9$
are thus triggered as $m = 3$.  
Similarly, 
we do not require distinct ID for each sensor. 
We just need a symbol size of $k = 3$, which is even smaller than that for 
$m = 2$.  
When $m$ increases, generally we can reduce $k$. We will discuss the factor of reduction in the following section.  


\begin{figure}[t]
    \centering
    \includegraphics[width=0.47\columnwidth]{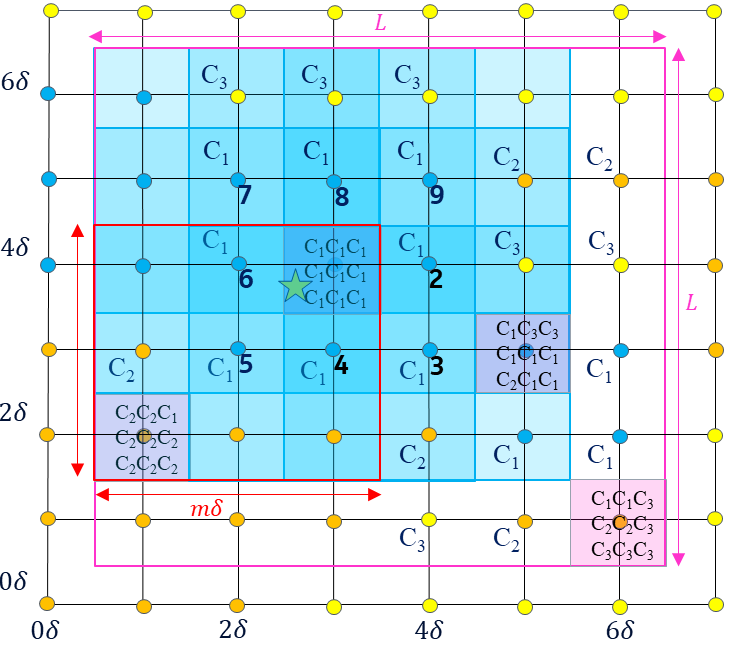} 
    \caption{$L = 6 \delta$ and $m = 3$. 
    Sensors are deployed for monitoring the same area.
    A sensor at $(i \delta, j \delta)$ would detect the presence of object in the  vicinity $x \in ((i-1.5) \delta, (i+1.5) \delta]$ and $y \in ((j-1.5) \delta, (j+1.5) \delta]$, and then transmits the pre-assigned ID to signal. 
    Here, sensors $1, 2, \ldots, 9$ are triggered as $m = 3 $.
    As shown above, only 3 distinct IDs (indicated as 3 colors) are required for the sensors.}   
    \label{fig:Fig3a} 
\end{figure}

\section{Decoding for Position Identification} 
\label{section:decoding}
In this section, we address the decoding problem associated with position identification based on the proposed MCGC scheme. 
For simplicity, we consider cyclic lattices.


Let $f$ be a color multiset code defined by a color mapping $\Phi\in\mathcal{C}_{\M;k}$.
By definition, if $f$ is $\m$-distinguishable, then $f$ is an one-to-one function and the inverse $f^{-1}$ defined on the image of $f$ exists.
This section focuses on characterizing the inverse function.

\subsection{The 1D Codes}

We first consider the case $m=2$ and $k$ is odd.
Assume that the number of colors is $k=2h+1$.
Following the discussion in Section~\ref{section:lower-bounds}, the $2$-distinguishable color multiset code $f$ is based on an Eulerian circuit of $G_{2h+1}$, the complete graph with a self-loop at each vertex.
The inverse $f^{-1}$ can potentially be characterized, provided that the underlying Eulerian circuit satisfies certain desirable properties.
The following provides a construction, including three steps, of an Eulerian circuit which has a nice property.
In this paper, we refer the Eulerian circuit obtained based on this construction to the \textit{canonical} Eulerian circuit.
\begin{enumerate}[(i)]
\item Consider $2h$ vertices labeled $1,2,\ldots,2h$, arranged cyclically on a circle in clockwise order.
Let $P_1$ be the path $1,2h,2,2h-1,3,2h-2,\ldots,h,h+1$, and for each $i$, $1\leq i\leq h$, define $P_i$ to be the path obtained from $P_1$ by a clockwise rotation by $i$ positions (equivalently, by applying the cyclic shift $(1,2,\ldots,2h)^i$ to its vertices).
\item Extend each path $P_i$ to $P'_i$ by attaching one additional edge to each endpoint, duplicating the endpoint labels, and then attaching the ending vertex to an extra vertex labeled $2h+1$.
\item Extend $P'_h$ by repeating the last vertex, $2h+1$, once. Finally, concatenate the resulting $h$ paths to yield a cycle that forms an Eulerian circuit.
\end{enumerate}

The length of each $P'_i$, $1\leq i\leq h-1$ is $2h+3$ and the length of $P'_h$ is $2h+4$.
Note that this construction basically follows the classic construction that the complete graph $K_{2h+1}$ can be decomposed into $n$ edge-disjoint Hamiltonian cycles. 
See~\cite[p.\,16]{Bollobas98}, for instance.

\begin{example}\label{ex:decode-2} \rm
Consider $h=4$.
The four paths $P_1,P_2,P_3$ and $P_4$ are shown in Fig.~\ref{fig:decoding}.
The second step yields $P'_1=1,1,8,2,7,3,6,4,5,5,9$, $P'_2=2,2,1,3,8,4,7,5,6,6,9$, and $P'_3=3,3,2,4,1,5,8,6,7,7,9$, while the last step further produces $P'_4=4,4,3,5,2,6,1,7,8,8,9,9$. 
The resulting canonical Eulerian circuit is
\begin{equation}\label{eq:decoding-9}
\begin{split}
& 1,1,8,2,7,3,6,4,5,5,9, \ 2,2,1,3,8,4,7,5,6,6,9, \\
& 3,3,2,4,1,5,8,6,7,7,9, \ 4,4,3,5,2,6,1,7,8,8,9,9.
\end{split}
\end{equation}
In Eqn.~\eqref{eq:decoding-9}, the blank spaces are inserted to distinguish between different paths.
\begin{figure}
\centering
\includegraphics[width=0.9\columnwidth]{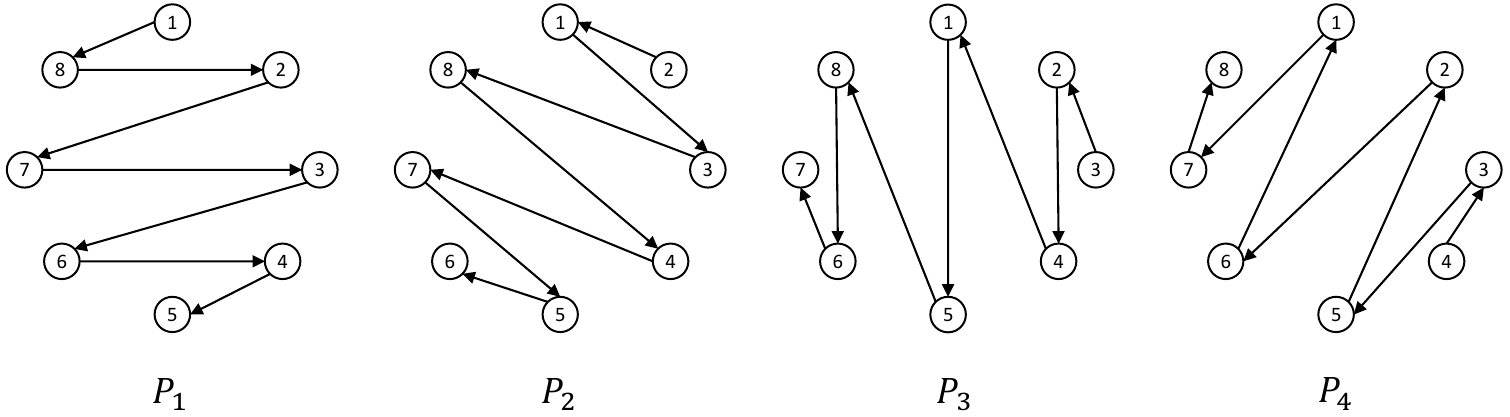}
    \caption{Four paths in the complete graph $K_8$.}  \label{fig:decoding} 
\end{figure} 
\end{example}

Among the paths $P_1,\ldots,P_h$ in the first step of the construction of the canonical Eulerian circuit, only $P_1$ may contain edges whose endpoint labels sum to $2h+1$ or $2h+2$; moreover, every edge in $P_1$ must have its endpoints summing to one of these two values.
This property allows us to determine which edge is under consideration by simply checking the sum of its endpoints.
More precisely, consider an edge with endpoints $a$ and $b$.
If $a=b$, of if one of $a$ or $b$ equals $2h+1$, then this edge appears on both sides of some path.
Otherwise, we consider the sum $a+b$.
If $a+b=2h+1$ or $2h+2$, then the edge belongs to $P'_1$, and we can further identify which edge it is by the smallest value of $\{a,b\}$ and the sum $a+b$.
Whenever $a+b$ is different from both $2h+1$ and $2h+2$, the edge must belong to $P_r$ for some $2\leq r\leq h$. 
This edge can be mapped to an edge $e$ in $P_1$ by a counterclockwise rotation of the $2h$ points, $1,2,\ldots, 2h$, by $(r-1)$ positions.
By construction, $e$ has endpoints sum $2h+1$ or $2h+2$.

Now, let $\Phi\in\mathcal{C}_{M,k}$ be the color mapping with the underlying graph is the canonical Eulerian circuit over $G_k$, where $k=2h+1$ and $M=(h+1)(2h+1)$.
Note that the number of edges of $G_k$ is $(h+1)(2h+1)$.
With the construction of the canonical Eulerian circuit, we have $\Phi(0)=1$, $\Phi(2h+1)=h+1$,
\begin{equation}\label{eq:decoding-canonical-1}
     \Phi(2i-1)=i \text{ and } \Phi(2i)=2h+1-i, \text{ for }1\leq i\leq h,
\end{equation}
due to the labels of the vertices in $P'_1$;
$\Phi(h(2h+3))=2h+1$ and
\begin{equation}\label{eq:decoding-canonical-2}
    \Phi(r(2h+3)-1)=2h+1 \quad \text{for } 1\leq r\leq h
\end{equation}
due to the labels of the ending vertices in each $P'_r$; 
and, for $0\leq a\leq 2h+1$,
\begin{equation}\label{eq:decoding-canonical-3}
    \Phi(r(2h+3)+a) = \Phi(a) + 1 \quad (\text{mod}\ 2h), \quad \text{for }1\leq r\leq h-1,
\end{equation}
due to the labels of the vertices in $P'_r$.
Note that the modulo operation in Eqn.~\eqref{eq:decoding-canonical-3} is taken over the set $\{1,2,\ldots,2h\}$.

Let $f$ be the color multiset code defined by $\Phi$.
The code $f$ is $2$-distinguishable, and it is essentially a function from $\mathcal{G}^c_M$ to $\mathcal{P}(2,2h+1)$, the collection of $2$-multi-subsets of $[2h+1]$.
The inverse $f^{-1}$ can be obtained from Eqns.~\eqref{eq:decoding-canonical-1} -- \eqref{eq:decoding-canonical-3}, as follows.
Let $\{a,b\}\in\mathcal{P}(2,2h+1)$ be a given multiset.
If $a=b$, then 
\begin{equation}\label{eq:decoding-1}
    f^{-1}(\{a,a\}) = \begin{cases}
        (a-1)(2h+3), & \text{if } 1\leq a\leq h, \\
        (a-h-1)(2h+3)+2h, & \text{if } h< a\leq 2h, \\
        (h+1)(2h+1)-2, & \text{if } a=2h+1.
    \end{cases}
\end{equation}
If exactly one of $a,b$ equals to $2h+1$, let $c=\min\{a,b\}$ and then
\begin{equation}\label{eq:decoding-2}
     f^{-1}(\{a,b\}) = \begin{cases}
         (h+1)(2h+1)-1, & \text{if } c=1, \\
         (c-1)(2h+3)-1, & \text{if } 1 < c\leq h, \\
         (c-h-1)(2h+3)+2h+1, & \text{if } h < c \leq 2h.
     \end{cases}
\end{equation}
Finally, if $a\neq b$ and non of $a,b$ is $2h+1$, we find the value $r$, $0\leq r\leq h-1$, such that 
\begin{equation}\label{eq:decoding-finding-r}
    a-r\ (\text{mod }2h) + b-r\ (\text{mod }2h) = 2h+1 \text{ or } 2h+2.
\end{equation}
Note that modulo operation in Eqn.~\eqref{eq:decoding-finding-r} is still taken over $\{1,2,\ldots,2h\}$.
Let $c=\min\{a-r\ (\text{mod }2h), b-r\ (\text{mod }2h)\}$ and $\textsf{sum}=a-r\ (\text{mod }2h) + b-r\ (\text{mod }2h)$.
Then,
\begin{equation}\label{eq:decoding-3}
     f^{-1}(\{a,b\}) = \begin{cases}
         r(2h+3)+2c-1, & \text{if } \textsf{sum}=2h+1, \\
         r(2h+3)+2c-2, & \text{if } \textsf{sum}=2h+2.
     \end{cases}
\end{equation}
The decoding procedure can be performed according to the algorithm~\ref{alg:decoding}.
This algorithm requires $O(h)$ time only at line $7$, while all other steps take constant time.
It should be noted that the size of the code is $(h+1)(2h+1)$, a quadratic function of $h$.

\begin{algorithm}[htb]
\caption{Decoding procedure of $f$ based on the canonical Eulerian circuit}
\label{alg:decoding}
\begin{algorithmic}[1] 
\REQUIRE A multiset $\{a,b\}$
\ENSURE A grid point $x=f^{-1}(\{a,b\})$ 
    \IF{$a=b$}
    \STATE $x$ is given in Eqn.~\eqref{eq:decoding-1};
    \ELSE 
        \IF{$a = 2h+1$ \textbf{or} $b = 2h+1$}
        \STATE $x$ is given Eqn.~\eqref{eq:decoding-2};
        \ELSE
            \STATE Find $r$, $0\leq r\leq h-1$, such that $a-r\ (\text{mod }2h) + b-r\ (\text{mod }2h) = 2h+1$ or $2h+2$; 
            \STATE $x$ is given in Eqn.~\eqref{eq:decoding-3};
        \ENDIF
    \ENDIF
\end{algorithmic}
\end{algorithm}

\medskip

\noindent\textbf{Example 6} (conti.)\textbf{.} 
Let $f$ be the color multiset code defined by the color mapping associated with the Eulerian circuit in~\eqref{eq:decoding-9}.
If the received multiset is $\{7,7\}$, by Eqn.~\eqref{eq:decoding-1}, the corresponding grid point is $(7-4-1)\cdot 11 + 8=30$.
If the received multiset is $\{4,9\}$, by Eqn.~\eqref{eq:decoding-2}, we have $c=4$ and the corresponding grid point is $(4-1)\cdot 11 - 1=32$.
If the received multiset is $\{2,4\}$, by Eqn.~\eqref{eq:decoding-finding-r}, we have $r=2$, $\textsf{sum}=10$ and $c=2$.
By applying Eqn.~\eqref{eq:decoding-3}, the corresponding grid point is $2\cdot 11 +2\cdot 2-2 =24$.

\medskip

It is worth emphasizing that even in the case $m=2$, if the given $2$-distinguishable sequence lacks a specific structure, the decoding process is rather difficult and can only be carried out by exhaustively searching through a look-up table.
This difficulty becomes even more pronounced in the more general case $m>2$.

\subsection{High-dimensional Codes based on Product Multiset Codes}

In the case of an arbitrary $\m$-distinguishable $n$-dimensional color multiset code, decoding is highly nontrivial and difficult.
In the following, we propose a feasible decoding strategy for the color multiset codes produced by product codes.

Following the setting in Section~\ref{section:product-code}, consider the $n$-dimensional color mapping $\Phi$ on $\mathcal{G}_{\M}$ with $\prod_{i=1}^n k_i$ colors given in Eqn.~\eqref{eq:product-code-def}: 
\begin{align*}
    \Phi(x_1,\ldots,x_n)=(\Phi_1(x_1),\ldots,\Phi_n(x_n)),
\end{align*}
which is based on $\Phi_i\in\mathcal{C}_{M_i;k_i}$ for $1\leq i\leq n$.
Let
\begin{align*}
    \theta:[k_1]\times[k_2]\times\cdots\times[k_n] \to [k_1k_2\cdots k_n]
\end{align*}
be a one-to-one correspondence.
Such a one-to-one correspondence can be established using a simple congruence property.
For example, when $n=2$, one may define $\theta((a,b))=k_1(a-1)+b$ so that $a$ and $b$ can be obtained by dividing $\theta((a,b))$ by $k_1$.

Let $f$ be a color multiset code defined by $\Phi$.
Assume that each color multiset code defined by $\Phi_i$, $1\leq i\leq n$, is $m_i$-distinguishable.
Then, by Proposition~\ref{prop:product-code}, it follows that $\Phi$ is $\m$-distinguishable.
Consider a color multiset, $S_{\m}(x_1,\ldots,x_n)$, tagged at $(x_1,\ldots,x_n)$.
We first consider the decoding procedure of the first coordinate, say $x_1$.
Decompose $S_{\m}(x_1,\ldots,x_n)$ into $m_1$ disjoint subsets by, for $j=0,1,\ldots,m_1-1$,
\begin{align*}
    A_j \triangleq \{\Phi(x_1+j,x_2+t_2,\ldots,x_n+t_n):\,0\leq t_i<m_i, 2\leq i\leq n\}.
\end{align*}
Under the one-to-one correspondence $\Phi$ with respect to the first coordinate, all elements in $A_j$ are projected to the same element $\Phi_1(x_1+j)$.
In this way, the multiset $S_{\m}(x_1,\ldots,x_n)$ is projected to the multiset $\{\Phi(x_1),\Phi(x_1+1),\ldots,\Phi(x_1+m_1-1)\}=S_{m_1}(x_1)$, each of which appears exactly $\prod_{i=2}^{n}m_i$ times.
Therefore, $x_1$ can be decoded in the $n$-dimensional case whenever it can be decoded in the color multiset code defined by $\Phi_1$.
This method can be applied to other coordinates, say the decoding of $x_2,\ldots,x_n$.
In other words, the proposed decoding strategy simplifies the high-dimensional case by considering each dimension individually.

\section{Color Coding Gain} 
\label{section:discussion}

By \textit{color coding gain} we refer to the factor of reduction of the number of bits used to label each sensor by the proposed protocol over the baseline reference protocol. 
In a 2D grid $\mathcal{G}_{M_1,M2}$, the number of bits needed to code an ID of a sensor under the baseline reference scheme is $\log_2{M_1M_2}$, while under the MCGC based protocol it is $\log_2 K_{M_1,M_2}(m_1,m_2)$.

In this section, we consider general $n$-dimensional grids.
Let $\M=(M_1,\ldots,M_n)$ and $\m=(m_1,\ldots,m_n)$ for some positive integers $M_1,\ldots,M_n,m_1,\ldots,m_n$ with $m_i<M_i$ for all $i$.
The color coding gain under the MCGC based protocol, denoted by $\mathcal{R}_{\M}(\m)$, is given by
\begin{equation}\label{eq:gain-definition}
    \mathcal{R}_{\M}(\m) \triangleq \frac{\log_2 K_{\M}(\m)}{\log_2 \left(\prod_{i=1}^{n}M_i\right)}.
\end{equation}

For ease of discussion, we use $\nu(\M)$ and $\nu(\m)$ to denote $\prod_{i=1}^n M_i$ and $\prod_{i=1}^n m_i$, respectively.
Let $k=K_{\M}(\m)$.
Consider an $\m$-distinguishable color multiset code defined by $\Phi\in\mathcal{C}_{\M;k}$.
By definition, the multisets $S_{\m}(x_1,\ldots,x_n)$ under $\Phi$ are all distinct for all $0\leq x_i\leq M_i-m_i$, $1\leq i\leq n$.
Similar to the argument in the proof of Proposition~\ref{prop:upper}, we represent each of these multisets as $\{1^{e_1},2^{e_2},\ldots,k^{e_k}\}$, where $e_s$ indicates the multiplicity of the element $s$ and therefore is a non-negative integer.
As an $\m$-block consists of $\nu(\m)$ grid points, we have $e_1+e_2+\cdots+e_k=\nu(\m)$.
By Eqn.~\eqref{eq:nonnegative-H}, the number of all possible color multisets $S_{\m}(x_1,\ldots,x_n)$ is $\binom{k+\nu(\m)-1}{\nu(\m)}$.
It follows that $\prod_{i=1}^{n}(M_i-m_i+1)\leq\binom{k+(\nu(\m))-1}{\nu(\m)}$, and thus
\begin{equation}\label{eq:gain-lower-proof_1}
    \nu(\M)=\prod_{i=1}^{n}M_i \leq \frac{k^{\nu(\m)}}{(\nu(\m))!}.
\end{equation}
Then, we have
\begin{align}
\log_2 (\nu(\M)) &\leq \nu(\m)\log_2{k}-\log_2(\nu(\m)!) \notag \\
&= \nu(\m)\log_2{k} -\big(\nu(\m)\log_2(\nu(\m))-\nu(\m)\log_2{e} +O(\log_2(\nu(\m))) \big), \label{eq:Stirling}
\end{align}
which implies that
\begin{equation}\label{eq:Stirling2}
\log_2{k} \geq \frac{\log_2(\nu(M))}{\nu(\m)} + \log_2(\nu(\m)) - \log_2{e} + \frac{c\log_2(\nu(\m))}{\nu(\m)},
\end{equation}
for some constant $c$.
The equality in Eqn.~\eqref{eq:Stirling} is due to \textit{Stirling's approximation formula}~\cite{Dutka91}.
Therefore, as $M_i$ goes to infinity, for each $i$, the color coding gain defined in Eqn.~\eqref{eq:gain-definition} has a natural lower bound given by
\begin{equation}\label{eq:gain-lower}
\mathcal{R}_{\M}(\m) \geq \frac{1}{\nu(\m)}.
\end{equation}

Now, we shall derive the color coding gain by means of the product multiset code.
By Proposition~\ref{prop:minimal}, $K_{\M}(\m) \leq \prod_{i=1}^{n}K_{M_i}(m_i)$, which implies by Eqn.~\eqref{eq:gain-definition} that 
\begin{equation}\label{eq:gain-PMC}
\mathcal{R}_{\M}(\m) \leq \frac{\sum_{i=1}^{n}\log_2 K_{M_i}(m_i)}{\sum_{i=1}^{n}\log_2{M_i}}.
\end{equation}

For given $M_i$ and $m_i$, let $k_i=K_{M_i}(m_i)$. 
By Corollary~\ref{corol:asymptotic}, we have $M_i\approx k_i^{m_i}$ as $M_i\to\infty$.
It follows from Eqn.~\eqref{eq:gain-PMC} that
\begin{equation}\label{eq:gain-PMC2}
\mathcal{R}_{\M}(\m) \leq \frac{\sum_{i=1}^{n}\log_2{k_i}}{\sum_{i=1}^{n}m_i\log_2{k_i}}.
\end{equation}

In particular, when $m=m_1=\cdots=m_n$, i.e., the case of $n$-dimensional hypercube detection blocks, we conclude from Eqn.~\eqref{eq:gain-lower} and Eqn.~\eqref{eq:gain-PMC2} that
\begin{equation}\label{eq:gain-square}
\frac{1}{m^n} \leq \mathcal{R}_{\M}(m,\ldots,m) \leq \frac{1}{m}.
\end{equation}

Note that we can configure the system parameter $m$ for the reduction factor at the cost of larger detection range of a sensor, as a kind of implementation tradeoff.   
For $m > 1$, one can also consider as a sensing collaboration or cooperative localization, which employs a larger detection range but could offer the reduction factor to the number of transmitted bits per node. 

In what follows, we will provide some experimental results.
We first consider the 1D cases.
Table~\ref{tab:K_M(m)-234} lists the minimum number of colors needed for $\mathcal{G}_M$, $M=10,100,1000,10000$, with square block size $m=2,3,4$, which can be obtained directly by Table~\ref{tab:M_m(k)-bounds} and Corollary~\ref{corol:Mcycle-m=23}.

\begin{table}[ht]
\centering
\begin{tabular}{|c|rrrr|}
\hline
\diagbox{$m$}{$M$} & $10$ & $100$ & $1000$ & $10000$  \\ \hline
2 & 5 & 15 & 45 & 141 \\ 
3 & 4 & 8 & 18 & 39 \\
4 & 3 & 6 & 11 & 21 \\ \hline 
\end{tabular}
\medskip
\caption{$K_M(m)$ for $M=10,100,1000,10000$ and $m=2,3,4$.}
\label{tab:K_M(m)-234}
\end{table}

Then, we can get the upper bounds of the color coding gains, given in Eqn.~\eqref{eq:gain-PMC}, for the 2D grid $\mathcal{G}_{M_1,M_2}$ cases with sizes $M_1,M_2\in\{10,100, 1000, 10000\}$ and the detection blocks $m_1\times m_2$ with $2\leq m_1,m_2\leq 4$.
For example, $K_{10000,100}(4,3)\leq K_{10000}(4)\times K_{100}(3)=21\times 8$, so $\mathcal{R}_{10000,100}(4,3)$ is upper-bounded by
\begin{align*}
    \frac{\log_2{21}+\log_2{8}}{\log_{2}{10000}+\log_2{100}}\approx 0.371.
\end{align*}
See Table~\ref{tab:K_MN(43)} for the color coding gain for each case with $m_1=4, m_2=3$.

\begin{table}[ht]
\centering
\begin{tabular}{|r|rrrr|}
\hline
\diagbox{$M_1$}{$M_2$} & $10$ & $100$ & $1000$ & $10000$  \\ \hline
10 & 0.539 & 0.460 & 0.433 & 0.413 \\ 
100 & 0.460  & 0.420 & 0.406 & 0.395 \\
1000 & 0.410 & 0.389 & 0.383 & 0.376 \\ 
10000 & 0.385 & 0.371 & 0.368 & 0.364 \\ \hline
\end{tabular}
\medskip
\caption{The color coding gains $\mathcal{R}_{M_1,M_2}(4,3)$ based on the product multiset code for $M_1,M_2\in\{10,100,1000,10000\}$.}
\label{tab:K_MN(43)}
\end{table}

Table~\ref{tab:R_M_N(m,m)-234} lists the color coding gain based on the product multiset code for each case with $2\leq m=m_1=m_2\leq 4$.
One can see that the values will slowly converge to $1/m$ as $M_1$ or $M_2$ goes to infinity.

\begin{table}[ht]
\centering
\begin{tabular}{|c|r|rrrr|}
\hline
$m$ & \diagbox{$M_1$}{$M_2$} & $10$ & $100$ & $1000$ & $10000$  \\ \hline
\multirow{4}{*}{$2$} & 10 & 0.698 & 0.625 & 0.588 & 0.569 \\ 
& 100 & 0.625  & 0.588 & 0.566 & 0.554 \\
& 1000 & 0.588 & 0.566 & 0.551 & 0.554 \\ 
& 10000 & 0.569 & 0.554 & 0.543 & 0.537 \\ \hline
\multirow{4}{*}{$3$} & 10 & 0.602  & 0.502 & 0.464 & 0.438 \\ 
& 100 & 0.502  & 0.451 & 0.432 & 0.415 \\
& 1000 & 0.464 & 0.432 & 0.418 & 0.406 \\ 
& 10000 & 0.438 & 0.415 & 0.406 & 0.397 \\ \hline
\multirow{4}{*}{$4$} & 10 & 0.477 & 0.418 & 0.379 & 0.360 \\ 
& 100 & 0.418  & 0.389 & 0.364 & 0.350 \\
& 1000 & 0.379 & 0.364 & 0.347 & 0.337 \\ 
& 10000 & 0.360 & 0.350 & 0.337 & 0.330 \\ \hline
\end{tabular}
\medskip
\caption{The color coding gains $\mathcal{R}_{M_1,M_2}(m,m)$ based on the product multiset code for $M_1,M_2\in\{10,100,1000,10000\}$ and $m=2,3,4$.}
\label{tab:R_M_N(m,m)-234}
\end{table}

\section{Conclusion} \label{section:conclusion}

We propose the concept of coding source data by mapping to alphabet multisets.  
A solution approach to this problem when the source data can be organized into an $n$-dimensional integer lattice or grid is introduced.
The solution, color multiset coding, is defined via a mapping from the grid to the alphabet, in which the elements are referred to as colors.  
The mapping construction for higher dimensional grids can be further decomposed to construction built on 1D grids using the idea of product multiset code.  
1D color codes that require the minimal number of colors is examined in details.
An example application of multiset coding to an object tracking problem on a proximity sensor network is presented to conclude the paper.

In the future, an important research direction is to explore more efficient color mappings for higher dimensional grids.
In particular, an interesting idea is to generalize the synthetic construction proposed in this paper.
Another direction is to consider real-world factors and challenges such as interference, measurement noise errors and engineering practice imperfections.
It is also worth noting that the proposed scheme can be generalized to multi-target tracking problem, which is also left for future study.

\bibliographystyle{IEEEtran}

\bibliography{IEEEabrv,reference,reference_vlp}
\end{document}